\documentclass[preprint,amsmath,amssymb,aps,prr,superscriptaddress]{revtex4-2}

\usepackage{amsmath, amsthm}
\usepackage{graphicx}
\usepackage{chemfig}

\theoremstyle{plain}
\newtheorem{thm}{Theorem}[section]
\newtheorem{prop}[thm]{Proposition}
\newtheorem{defn}[thm]{Definition}

\theoremstyle{definition}

\newtheorem{fact}[thm]{Fact}

\newcommand\sumstack[2]{
\substack{
#1\\
#2}
}
\newcommand\tristack[3]{
\substack{
#1\\
#2\\
#3}
}

\newcommand\eyej{
\substack{
i\\
i\neq j}
}
\newcommand\eyek{
\substack{
i\\
i\neq k}
}
\newcommand{\la}{\langle}
\newcommand{\ra}{\rangle}
\newcommand\adj{{\rm adj}}
\newcommand\mL{\mathcal{L}}
\newcommand\proj{{\rm proj}}
\newcommand\sij{{\mathcal{L}^{ij}}}
\newcommand\norm[1]{\left\lVert#1\right\rVert}

\newcommand\wpp{{\omega_p}}

\newcommand\eei{{e^{\epsilon_i}}}

\newcommand\eps{\epsilon}
\newcommand\oom{\omega}
\newcommand\gam{\gamma}

\newcommand\epp{{e^{\epsilon_p}}}
\newcommand\epl{{\epsilon_p}}
\newcommand\eil{\epsilon_i}

\begin{document}

\title{Processive and Distributive Non-Equilibrium Networks Discriminate in Alternate Limits}
\author{Gaurav G. Venkataraman}
\email{Correspondence:  gauravvman@gmail.com or dj333@cam.ac.uk}
\affiliation{Wellcome/CRUK Gurdon Institute, University of Cambridge.\\Tennis Court Rd, Cambridge, CB2 1QN, UK.}
\affiliation{Division of Medicine, University College London.  London WC1E 6BT, UK.}

\author{Eric A. Miska}
\affiliation{Wellcome/CRUK Gurdon Institute, University of Cambridge.\\Tennis Court Rd, Cambridge, CB2 1QN, UK.}
\affiliation{Department of Genetics, University of Cambridge.  Downing Street, Cambridge CB2 3EH, UK.}
\affiliation{Wellcome Sanger Institute, Wellcome Genome Campus, Cambridge CB10 1SA, UK.}

\author{David J. Jordan}
\email{Correspondence: gauravvman@gmail.com or dj333@cam.ac.uk}
\affiliation{Wellcome/CRUK Gurdon Institute, University of Cambridge.\\Tennis Court Rd, Cambridge, CB2 1QN, UK.}
\affiliation{Department of Genetics, University of Cambridge.  Downing Street, Cambridge CB2 3EH, UK.}

\begin{abstract}
We study biochemical reaction networks capable of product discrimination inspired by biological \emph{proofreading} mechanisms.  At equilibrium, product discrimination, the selective formation of a ``correct'' product with respect to an ``incorrect product'', is fundamentally limited by the free energy difference between the two products.  However, biological systems often far exceed this limit, by using discriminatory networks that expend free energy to maintain non-equilibrium steady states.  Non-equilibrium systems are notoriously difficult to analyze and no systematic methods exist for determining parameter regimes which maximize discrimination.  Here we introduce a measure that can be computed directly from the biochemical rate constants which provides a condition for proofreading in a broad class of models, making it a useful objective function for optimizing discrimination schemes.  Our results suggest that this measure is related to whether a network is processive or distributive.  Processive networks are those that have a single dominant pathway for reaction progression, such as a protein complex that must be assembled sequentially. while distributive networks are those that have many effective pathways from the reactant to the product state; e.g. a protein complex in which the subunits can associate in any order.  Non-equilibrium systems can discriminate using either binding energy  (\emph{energetic}) differences or activation energy (\emph{kinetic}) differences. In both cases, proofreading is optimal when dissipation is maximized.  In this work, we show that for a general class of proofreading networks, energetic discrimination requires processivity and kinetic discrimination requiring distributivity.  Optimal discrimination thus requires both maximizing dissipation and being in the correct processive/distributive limit.  Sometimes, adjusting a single rate may put these requirements in opposition and in these cases, the error may be a non-monotonic function of that rate.   This provides an explanation for the observation that the error is a non-monotonic function of the irreversible drive in the original proofreading scheme of Hopfield and Ninio. Finally, we introduce mixed networks, in which one product is favored energetically and the other kinetically.  In such networks, sensitive product switching can be achieved simply by spending free energy to drive the network toward either the processive limit or the distributive limit.  Biologically, this corresponds to the ability to select between products by driving a single reaction without network fine tuning.  This may be used to explore alternate product spaces in challenging environments.
\end{abstract}

\maketitle
\section{Introduction} 
Chemical systems in isolation will evolve toward thermodynamic equilibrium, a unique steady state where the concentrations of chemical species no longer change with time, no entropy is produced, and the relative concentrations of different species are a function of their free energy differences alone.  In biology, thermodynamic equilibrium is synonymous with death and biochemical systems must avoid it by continuously using energy to maintain non-equilibrium steady states.  Far from equilibrium, state occupancies are no longer a function of free energies, in fact, consistent free energies cannot be assigned to out of equilibrium states, and can in principle depend on the full details of every transport process and chemical reaction rate in the system.  This allows for much greater flexibility in systems far from equilibrium.  This freedom comes at a cost; such systems rarely permit closed form, analytic solutions for quantities of interest, such as steady state concentrations, chemical fluxes, or entropy production (dissipation)\cite{Schnakenberg1976}.  

The difference between equilibrium and non-equilibrium thermodynamics is especially salient in the problem of biological discrimination.  Discrimination refers to the increase in the concentration of one ``correct'' product, relative to another ``incorrect'' product.  The ability of living systems to process and transmit information reliably depends on the accuracy of its biochemical reactions; this accuracy can be quantified as the ratio of these ``correct'' and ``incorrect'' products.  If a discriminatory system were at equilibrium, this ratio would be fundamentally limited by the free energy difference of the two products; discrimination beyond that would be impossible.  However, as first noted by Hopfield \cite{Hopfield1974} and Ninio \cite{Ninio1975}, biological processes show accuracy far beyond this limit.  For example, in DNA replication, error rates of $\sim10^{-9}$; are observed while the equilibrium limit is $\sim10^{-4}$.  Hopfield and Ninio proposed a system that they called ``kinetic proofreading'', which, by coupling certain reactions to an external chemical potential via the hydrolysis of ATP (e.g.), drives the system out of equilibrium, thus negating the limit and permitting enhanced discrimination.

In a simple chemical reaction, the concentration of the final product is determined by its free energy relative to the reactants, while the rate of the reaction is determined by the activation energy barrier and the systems temperature.  Activation energy differences are independent of free energy differences and thus cannot contribute to discrimination at thermodynamic equilibrium.  However, once a system is driven out of equilibrium, activation energy differences can also be used to discriminate \cite{Bennett1979,Bennett1982}.  In what follows, we will distinguish these two types of non-equilibrium discrimination, using energetic discrimination to refer to discrimination based on binding energy differences and kinetic discrimination to refer to that based on activation energy differences following Sartori and Pigolotti \cite{Sartori2013,Sartori2015}.  

In cells, biological processes are often carried out by heterogenous, multi-component complexes.  Such complexes are ubiquitous in biological information processing systems, from the protein translation system in the ribosome \cite{Staley2009}, to the gene regulatory networks that control and carry out transcription \cite{Hnisz2017}. Multicomponent complex formation may be a mechanism for assuring the accuracy of biological processes, as discrimination can be enhanced in reaction schemes with many intermediate steps \cite{Murugan2014}.  Reactions such as these, where many intermediate complexes are formed on the way to the final product, may proceed either processively or distributively.  Processive reactions must travel a single dominant path from reactants to products, while distributive reactions can go from reactants to products in many ways.  For example, the complex formation shown in Figure \ref{fig:complex_formation}(a) shows a processive mechanism where the association between the components must occur in order, as the nested nature of the molecular shapes ensures that the binding of later components requires that earlier sub-complexes have already formed.  Contrast this to distributive complex formation, as shown in Figure \ref{fig:complex_formation}(b), where the components of the complex are free to associate in almost any order.  This gives many effective pathways  along which complex formation can occur.  Chemical networks need not rely on unique molecular properties such as the shapes shown in Figure \ref{fig:complex_formation}(a) to realize processive or distributive assembly.  In fact, a single network, such as the ladder like network shown in Figure \ref{fig:complex_formation}(c,d), can be either processive or distributive depending on the rate constants.  For example, a complex might form around an enzyme (E) and its substrate (S), which can exist in either a modified form ($\star$) or an unmodified form and which associates with many complementary subunits ($C_n$). The modification in this example could be phosphorylation e.g., with the up and down reactions being carried out by a phosphatase and a kinase respectively. In the network in Figure \ref{fig:complex_formation}(c), removal of the modification (red arrow) are rare for the intermediate sub-complexes and modified complexes cannot participate in the final reaction so must dissociate completely and reform.  Thus, there is only a single dominant path to the end state, along the top of the ladder Figure \ref{fig:complex_formation}(c, shown in green).  In this context, modification events can be viewed as ``catastrophes'' Figure \ref{fig:complex_formation}(c, dashed blue arrows) , requiring complete disassembly of the complex. If these catastrophes are more likely to occur for an incorrect substrate than for the correct one, such a network can form the basis of a highly selective discriminatory scheme \cite{Murugan2012}.  

In contrast, consider what happens if the rate of removal of the modification is greatly increased [Fig. \ref{fig:complex_formation}d].  In this case, the complex is free to form with either the unmodified or the modified substrate, as the network can easily move from the lower path of the ladder to the upper path at any time.  A typical trajectory [Figure \ref{fig:complex_formation}(d, blue arrows)] might involve the complex forming partially with the modified substrate and then have the modification removed, after which the reaction can proceed along the upper path.  In this network, this shift can happen at any intermediate and the reaction can even switch from top to bottom  multiple times, giving many effective paths [Fig. \ref{fig:complex_formation}(d, green)] to the final product. The change in the rate of removal of the modification reaction [Figure \ref{fig:complex_formation}(c and d, red arrow)] could be the result of many things. For example, in the case when the modification is phosphorylation, this rate could be increased by increasing the expression of the phosphatase.  

\begin{figure*}[htb]
\includegraphics[width = 0.9\textwidth]{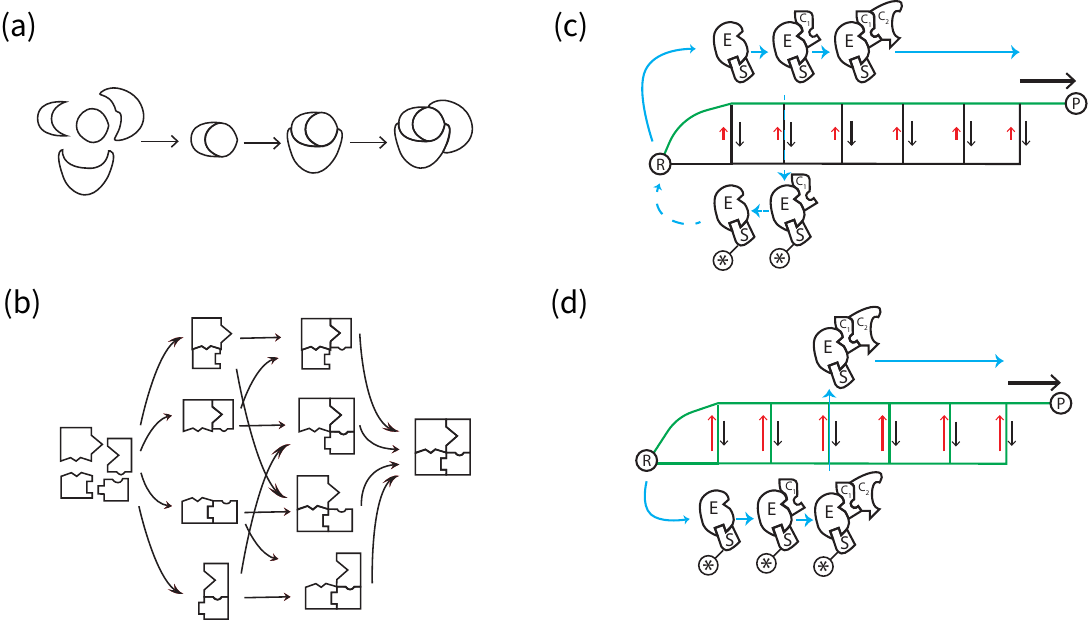}
\caption{\textbf{Processive and Distributive Reaction Networks}  Two examples of protein complex formation. In (a) reactions must occur sequentially resulting in a processive network, with a single path for assembly.  Here, nested molecular shape confers the processivity by dictating an order in which the molecules must assemble.  In contrast, in the distributive network shown in (b), the subunits are free to associate in almost any order allowing many effective paths for assembly. Processivity and distributivity need not be conferred by unique molecular properties such as shape, in fact, reactions which use the same components may be processive or distributive based on the reaction rates. For example, consider a ladder topology network. This network can be changed from a processive network (c) to a distributive network (d) by increasing the rate of a single reaction, the ``up'' reaction (red arrow). In this example, the top path differs from the bottom path by the addition of a modification ($\star$) to (S).  In (c), the modification is very rarely removed, and the modified complex cannot complete the reaction; any modification will be a ``catastrophe'' (c, blue path, dashed) requiring the complex to disassemble and start over before finally completing (c, blue path, solid).  Thus, all successful reactions must follow a single pathway from reactants (R) to products (P) (c, green path).  However, if the rate of the red reaction is greatly increased, the modification can be removed easily at any step.  Thus the complex can form either along the top path or along the bottom path and can switch at any time like in the example trajectory (d, blue path), giving many effective paths (d, green paths).
\label{fig:complex_formation}}
\end{figure*}

In this work we introduce a measure that quantifies the degree of distributivity versus processivity in a network.  This measure is a global property of the network that depends on both network topology and the reaction rates between states.  We show that distributivity is required for out of equilibrium networks to discriminate based on activation energy differences, and that processivity is required to discriminate based on binding energy differences.  We call this measure orthogonality because it precisely quantifies the degree to which the columns of the graph Laplacian are mutually orthogonal to one another.  We use this measure to solve an outstanding question about the non-monotonic behavior of the discrimination ratio in response to increasing dissipation in a classic proofreading scheme.  In spirit of previous work  \cite{Ehrenberg1980,Murugan2012,Murugan2014} that sought to study how systems can exist in different non-equilibrium regimes without changing the network topology, we explore discrimination in two classes of fixed topologies, the so-called ``butterfly'' \cite{Wong2018-ys} and ``ladder'' \cite{Murugan2012} graphs.  We show how the rate constants of these discrimination schemes can be tuned to put the network into either processive or distributive regimes, and thus be utilized for energetic or kinetic discrimination respectively.  Finally we show that we can design networks in which orthogonality is extremely sensitive to changes in a single reaction potential and demonstrate a principled way to design systems that can switch from one product forming regime to another without changing the ``hard-wiring'' of the system, and without extensive ``fine-tuning'' of chemical potentials throughout the system, but rather by modulating a single chemical drive.

The ability of a chemical system to change its product space simply by changing the availability of a chemical driving force, such as ATP, provides interesting ways in which biological systems might respond to environmental conditions.  Living systems spend a large proportion of their energy on maintaining osmolarity and membrane potentials through the actions of ATP-driven pumps \cite{Albe_2002_book}.  The idea that modulation of ATP availability could drive a chemical reaction network from one product space to another raises interesting possibilities for mechanisms of either improvisation or contingency in response to adversity.  

\section{Preliminaries}
We consider systems whose dynamics are described by continuous time Markov chains.  System states and transitions can be represented as a strongly connected, directed graph with $n$ states, and have dynamics represented by a matrix differential equation known as the Master equation
\[
\frac{d\mathbf{p}}{dt} = \mathcal{L}\mathbf{p}
\]
where $\mathcal{L}$ is the $n\times n$ {\it Laplacian matrix}, also known as the \emph{generator} in stochastic thermodynamics.  This matrix encodes the transition rates $k_{ij} = (j\to i)$ of the network in its off-diagonal elements $(i\neq j)$.  The diagonal elements are chosen such that all columns sum to zero:
\[
\mathcal{L}_{ij} =
 \begin{cases}
       {\hspace{4mm}k_{ij}} & i\neq j, \ k_{ij}\ge0\\
       {-\sum_j k_{ij}} & i=j,\\
\end{cases}
\]
and $\mathbf{p}$ is an $n$ dimensional vector representing the dynamic occupancy of the network states. We require these systems to be strongly-connected, meaning that any state is accessible from any other state, though not necessarily directly. Thus there are no ``absorbing'' states or isolated subgraphs.  Such networks have a single unique steady state \cite{Mirzaev2013}, which is the solution $\rho$ to the equation $\mathcal{L}\rho=0$.  Mathematically, this vector $\rho$ is called the \emph{nullspace}, or \emph{kernel}, of the Laplacian $\mathcal{L}$; it is also the eigenvector of the $\mathcal{L}$ corresponding to the eigenvalue of $0$.  Physically, $\rho$ is of interest because it is the (possibly non-equilibrium) steady state of the network.  In equilibrium thermodynamics, where detailed-balance holds, $\rho$ can be solved for exactly, as the ratio of the steady state concentrations of any two species $i$ and $j$ can be computed directly from their free energy difference.  Out of equilibrium, however, energies of states are not well defined \cite{Murugan2014}, and the calculation of $\rho$ in such networks does not generally permit a simple analytic solution. Consider, for example, the network with rate constants given in Figure \ref{fig:polytope_geometry}(a).  Based on these rate constants, no consistent free energies can be assigned to the three states as this system is not in equilibrium, and energy is dissipated in each cycle.  The free energy differences (up to an independent multiplicative factor) of states A relative to B is $\log{2}$,  and of B relative to C is $\log{1.5}$, however, the free energy of state C relative to A is not $\log{3}$, as would be expected from equilibrium considerations, but rather it is equal to 0.  Thus, there is free energy of $\log(3)$ dissipated for each cycle around this network.  For general non-equilibrium networks it is more difficult to compute the dissipation directly as we have done here, however, as long as the steady state vector $\rho$ is known, we can calculate the dissipation from $\rho$ and the rate constants $k_{ij}$ as the entropy production rate \cite{Schnakenberg1976,Hill2005-zc}, 
\begin{equation}
 \dot{S_i} = \frac{1}{2}\sum_{i,j}(k_{ij}\rho_j-k_{ji}\rho_i)\ln{\frac{k_{ij}\rho_j}{k_{ji}\rho_i}}
 \label{eq:dissipation}
\end{equation}
and when we refer to ``dissipation'' in what follows we will be calculating it according to Equation \ref{eq:dissipation}.  
\subsection{Laplacian Geometry}

\begin{figure*}[htb]
\includegraphics[width = 0.9\textwidth]{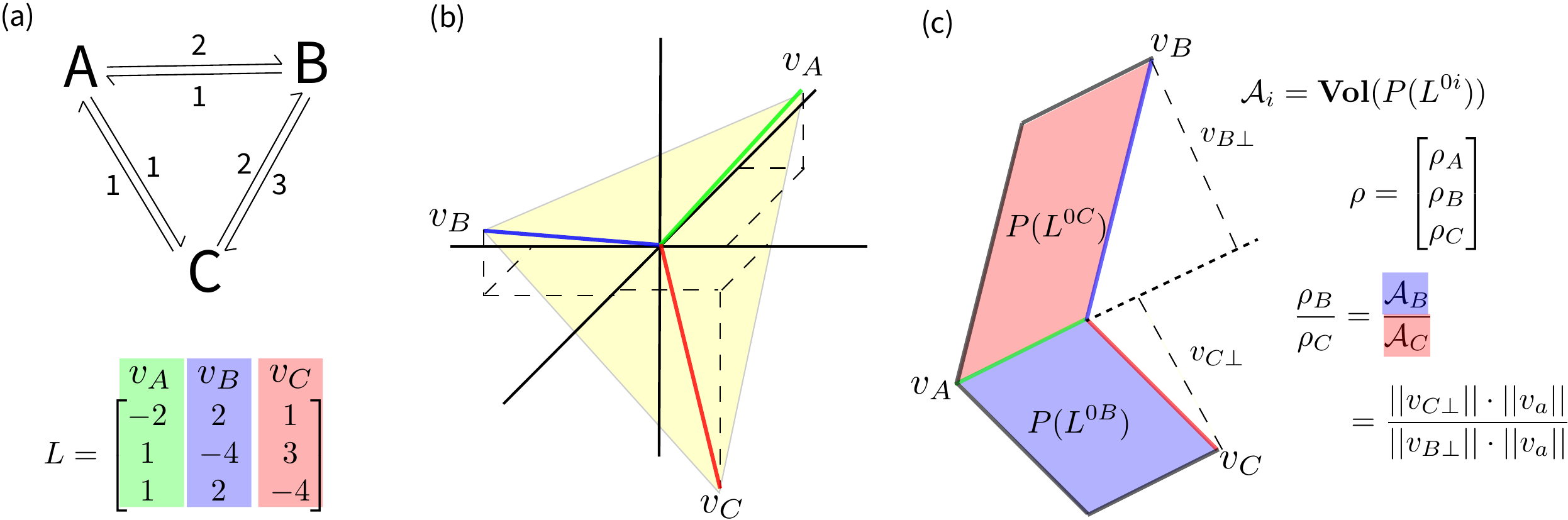}
\caption{A general reaction scheme can be shown as directed graph such as the 3-state reaction network shown in (a).  A Laplacian matrix (L) takes the rate constants from the directed graph and the rate constant from state $i$ to $j$, $k_{i\to j}$ is in the $i$ith row and $j$th column of $L$ and the diagonal elements are the negative of the column sum. Each of the columns of $L$ is a vector, here there are three vectors $v_A$ (green), $v_B$ (blue) and $v_C$ (red).  These vectors live in $\mathcal{R}^3$, but only span $\mathcal{R}^2$. The three vectors are shown in (b) with the same colors as (a) as well as the plane that they span shown in yellow.  We can visualize the vectors in the span of $L$ with an appropriate projection (c).  The polytope associated with $P(L^{0C})$ is the parallelogram with sides $v_A$ and $v_B$ (blue) and the polytope associated with $P(L^{0B})$ is the parallelogram with sides $v_A$ and $v_C$.  These polytopes share the facet $P(L^{BC})$ which is the vecotr $v_A$ in this case. The ratios of the elements of $\rho$, the steady state solution to $L\rho=0$ is given by the ratio of the areas of these two polytopes, and in this example, $\frac{rho_B}{\rho_C}$ is given as the area of the red parallelogram divided by the ratio of the blue parallelogram.   The area of $P(L^{0B})=\sqrt{147} = 7\sqrt{3}$ and the area of $P(L^{0C})=\sqrt{108} = 6\sqrt{3}$, giving a ratio $\frac{\rho_B}{\rho_C}=7/6$.  This can also be calculated using the base height formula, choosing the shared facet $P(L^{BC})=v_A$ as the base.  The ratio can then be computed simply as the ratio of the perpendicular components.
\label{fig:polytope_geometry}}
\end{figure*}

For each directed graph, such as the one shown in the upper panel of Figure \ref{fig:polytope_geometry}(a), there is a corresponding Laplacian $\mathcal{L}$, [Fig. \ref{fig:polytope_geometry}(a), lower], where $\mathcal{L}$ is an ($n\times n$) with rank $(n-1)$ \footnote{This is due to the graph being strongly-connected and the normalization condition}. The columns of $\mathcal{L}$ are $n$ vectors in an $n$ dimensional space, however, these vectors span only an $n-1$ dimensional subspace, as the matrix is not full rank.  This is shown in Figure \ref{fig:polytope_geometry}(b) where the three vectors $v_A$, $v_B$, and $v_C$ span a plane. The nullspace, or kernel, of $\mathcal{L}$ is the vector that gives the steady steady solution $\rho$ to the matrix differential equation $\frac{d\mathbf{p}}{dt}=\mathcal{L}\rho=0$.  In this example, the dimension of $span(\mathcal{L})$ is 2 and we can visualize the three vectors $v_A$, $v_B$, and $v_C$ in the plane with an appropriate coordinate transform.  In general, the matrix $B = \{\mathbf{e_1}-\mathbf{e_2},\mathbf{e_2}-\mathbf{e_3},...,\mathbf{e_{(n-1)}}-\mathbf{e_n}\}$, where $\mathbf{e_i}\in\mathcal{R}^n$ is the standard basis, provides a basis for $\mathcal{L}$ in $\mathcal{R}^{(n-1)}$ \cite{meyer2018}.  While this projection is useful for visualization, it is not part of the approximation we will introduce, nor is it required for the measure we present later.  Consider the three vectors $v_A$ (green), $v_B$ (blue) and $v_C$ (red) in Figure \ref{fig:polytope_geometry}(c).  If we remove the $i$th vector $v_i$, equivalent to removing the $i$th column of $\mathcal{L}$, then the remaining vectors form a shape called a polytope, which in two dimensions in a parallelogram and in three, a parallelepiped.  If we call $\mathcal{L}$ with the $i$th column removed $\mathcal{L}^{0i}$, then the polytope associated with state $i$ will be denoted $P(\mathcal{L}^{0i})$. The polytopes associated with B, $P(\mathcal{L}^{0B})$, consisting of $v_A$ and $v_C$, and that associated with C, $P(\mathcal{L}^{0C})$, consisting of $v_A$ and $v_B$ are shown in Figure \ref{fig:polytope_geometry}(c) in red and blue respectively.  The geometric insight of this paper is that the Laplacian matrix $\mathcal{L}$ defines a collection of polytopes $P(L^{0i})$ associated with the states $i$, and that the ratio of any two steady state concentrations is given by the ratio of the volumes of their corresponding polytopes. (A proof of Equation \ref{eq:polytope_volume} is given in Appendix \ref{app:ratio_volume_proof}).  
\begin{equation}
\frac{\rho_i}{\rho_j}=\frac{Vol(P(L^{0i}))}{Vol(P(L^{0j}))}
\label{eq:polytope_volume}
\end{equation}
In discrimination schemes, it is often the ratio of the steady state concentrations for the correct and the incorrect product that is of interest. Thus, the problem of computing this ratio reduces to a problem of computing the ratio of the volumes of the two polytopes as given by Equation \ref{eq:polytope_volume}. 
\subsection{Polytope Volumes}
Geometrically, the volume of a polytope is given in general by the famous ``base-height'' formula, that is, the volume $Vol(P(\mathcal{A}))$ is given $\|a_i\|\cdot Vol(P(\mathcal{A}^{i}))$ where $P(\mathcal{A}^{i}$) is the polytope formed from $\mathcal{A}$ with the $i$th column removed and $\|(a_i)\|$ is the magnitude $a_i$, the component of the $v_i$ perpendicular to $span(\mathcal{A}^{i})$.  In two dimensions the height is the perpendicular component of the adjacent side with respect to the chosen base. Note we are free to choose either side as the ``base''. In three dimensions, we have a parallelepiped comprised of three faces.  The volume can be given as the area of any of these faces multiplied by the magnitude of perpendicular component of the remaining side with respect to the chosen face. Any face of a parallelepiped is itself a parallelogram, and its area can therefore also be computed using the base height formula for the remaining face.  This procedure generalizes to higher dimensions were we choose one column, $v_i$, of $\mathcal{A}$ and compute its height $\|a_i\|$ perpendicular to the subspace spanned by $\mathcal{A}^i$.  We can then compute the volume of the base in the same way, in one fewer dimensions, iteratively until we reach the final one-dimensional subspace.  This gives the general formula,
\begin{eqnarray*}
Vol(P(\mathcal{A})) &=&  \|a_i \|Vol(P(\mathcal{A}^{i}))\\
&=&  \|a_i \|\left(\|a_j\|Vol(P(\mathcal{A}^{ij})\right)\\
&=&  \|a_i \|\|a_j\|\left(\|a_k\|Vol(P(\mathcal{A}^{ijk})\right)...\\
 &=& \prod_i \|a_i \|
\end{eqnarray*}
%Alternatively, we could compute the volume in terms of determinants.  \[ Vol(P(\mathcal{A}))) = \sqrt{Det(\mathcal{A}^\top\mathcal{A})}\] where $Det[\cdot]$ is the determinant and $A^\top$ is the transpose of $A$ \cite{Drucker2015}. Equation \ref{eq:polytope_volume} can then be rewritten as, 
%\begin{equation}
%\frac{\rho_i}{\rho_j}= \sqrt{\frac{Det[(\mathcal{L}^{0i})^\top(\mathcal{L}^{0i})]}{Det[(\mathcal{L}^{0j})^\top(\mathcal{L}^{0j})]}}\
% \label{eq:determinant_volume}
%\end{equation}
\subsection{Ratio of Polytope Volumes}
In general, an $n$ dimensional polytope will have $n$ facets, each $(n-1)$ dimensional, which are the higher dimensional equivalent of faces.  The polytope $P(\mathcal{L}^{0i})$ corresponding to state $i$ and the polytope $P(\mathcal{L}^{0j})$ corresponding to state $j$ will always share a facet which is formed from $\mathcal{L}^{ij}$ which is $\mathcal{L}$ with \emph{both} columns $i$ and $j$ removed.  In the example in Figure \ref{fig:polytope_geometry}, this the polytopes are 2-dimensional and their shared facet is simply $v_A$, and in general, it will be an $(n-2)$-dimensional polytope.  In the example in Figure \ref{fig:3d_polytopes} the 4-state Laplacian defines four three-dimensional  polytopes, which each share a two-dimensional facet (shown in red).  That the polytopes associated with $i$ and with $j$ share a facet comes from the fact that we can remove columns in any order.  We have already removed column $i$ and column $j$ to obtain the polytopes associated with state $i$ and state $j$ respectively. To calculate the volumes of these polytopes using the base-height formula, we will start by using removing $v_j$ from $P(\mathcal{L}^{0i})$ and $v_i$ from $P(\mathcal{L}^{0j})$ giving the same ``base'' for both, namely $P(\mathcal{L}^{ij})$, which is their shared facet.  This leads to a simplification of Equation \ref{eq:polytope_volume},
\begin{eqnarray}
\frac{\rho_i}{\rho_j}&=&\frac{Vol(P(L^{0i}))}{Vol(P(L^{0j}))}\nonumber\\
&=&\frac{\|(v_j)_\bot\|Vol(P(L^{ij}))}{\|(v_i)_\bot\|Vol(P(L^{ij}))}=\frac{\|(v_j)_\bot\|}{\|(v_i)_\bot\|}
\label{eq:height_ratio}
\end{eqnarray}
Where $(v_i)_\bot$ is the component of $v_i$ which is perpendicular to the $span(\mathcal{L}^{ij})$.  In the example shown in Figure \ref{fig:polytope_geometry}, by removing $v_b$ we get the polytope associated with $B$, $P(L^{0B})$ shown in red, similarly removing $v_C$ gives the polytope associated with $C$ shown as the parallelogram $P(L^{0C})$ in blue.  We can see that these polytopes share the facet $L^{BC}$ which is $L$ with columns $v_B$ and $v_C$ removed, which is simply $v_A$.  In this example, choosing the shared base to be $v_A$, the ratio of the areas is $\frac{\rho_B}{\rho_C}=\frac{\|v_{C\bot}\|\|v_A\|}{\|v_{B\bot}\|\|v_A\|}$.
\subsection{The Discrimination Ratio in Terms of Projections}
The component of $v_i$ perpendicular to the $span(\mathcal{L}^{ij})$ can be written in terms of the subspace projection of $v_i$ onto $span(\mathcal{L}^{ij})$,  as
\[
(v_i)_\bot = v_i-{\rm proj}_{\mathcal{L}^{ij}}(v_i)
\]
This is analogous to decomposing a vector into parallel and perpendicular components, $v_i = v_{i\bot}+v_{i\parallel}$, where $v_{i\parallel} = {\rm proj}_{\mathcal{L}^{ij}}(v_i)$. Therefore, Equation \ref{eq:height_ratio} can be rewritten as
\begin{equation}
\frac{\rho_i}{\rho_j} = \frac{\|v_j - {\rm proj}_{\mathcal{L}^{ij}}(v_j)\|}{\|v_i - {\rm proj}_{\mathcal{L}^{ij}}(v_i)\|}.
\label{eq:projection_ratio}
\end{equation}
In the example shown in Figure \ref{fig:polytope_geometry}, the subspace $L_{BC}$ is 1-dimensional and projections onto it are simple to compute in terms of the vectors $v_i$.  They are given explicitly as 
\[{\rm proj}_{L^{BC}}(v_i) = {\rm proj}_{v_A}(v_i) = \frac{\langle v_i,v_a \rangle}{\|v_a\|}\]. 
In this example, we can calculate the ratio of the occupancy of $B$ to $C$ at steady state directly,
\[\frac{\rho_B}{\rho_C} = \frac{\|v_C - {\rm proj}_{v_A}(v_C)\|}{\|v_B - {\rm proj}_{v_A}(v_B)\|}=\frac{\|v_C-\frac{\langle v_C,v_A\rangle}{\|v_A\|}\|}{\|v_B-\frac{\langle v_B,v_A\rangle}{\|v_A\|}\|}\] which for the values shown [Fig \ref{fig:polytope_geometry}a] give a ratio of 7/6.  Similarly, $\frac{\rho_B}{\rho_A}=7/10$ and $\frac{\rho_C}{\rho_A}=6/10$.  If we combine these ratios with the normalization condition that $\sum\rho=1$, $\rho$ is given as $[10/23, 7/23, 6/23]$.  Again, note that because this system is not in equilibrium, detailed balance does not hold, it is easy to see e.g. that $\rho_Ak_{A\to C} \neq \rho_Ck_{C\to A}$.
\subsection{An Approximation of the Discrimination Ratio}
Equation \ref{eq:projection_ratio} shows that an analytical expression for the projection onto the subspace spanned by $\mathcal{L}^{ij}$ will yield an analytic expression for the discrimination ratio.  However, computing such a projection requires having an ortho-normal basis for the subspace.  If we have such an ortho-normal basis $\mathcal{L}^{ij}_{\text{orth}}$, then the projection is given simply by,
\begin{equation}\label{eq:simple_projection}
{\rm proj}_{\mathcal{L}^{ij}_{\text{orth}}}(v_i) = \sum_{l\in \mathcal{L}^{ij}}\langle v_i, v_l \rangle v_i.
\end{equation}
However, the columns of $\mathcal{L}^{ij}$ will not, in general, form an ortho-normal basis.  We can orthogonalize the subspace, using a procedure such as the Gram-Schmidt process or by computing a matrix inverse e.g., however the recursive nature of these procedures yields expressions that are generally not analytically tractable. However, a general  solution for the discrimination ratio can be derived if these projections can be computed simply.  In the special case when the columns of $\mathcal{L}^{ij}$ are orthogonal, if we normalize the columns to unit length, and denote the resulting matrix $\widehat{\mathcal{L}}^{ij}$, and compute the ratio with the projection given in Equation \ref{eq:simple_projection}, yielding the expression,
\begin{equation}
\frac{\rho_i}{\rho_j} = \frac{\|v_j - \sum_{l\in \widehat{\mathcal{L}}^{ij}}\langle v_j, \widehat{v_l} \rangle v_j\|}{\|v_i -  \sum_{l\in \widehat{\mathcal{L}}^{ij}}\langle v_i, \widehat{v_l} \rangle v_i\|}.
\label{eq:projection_ratio_orth}
\end{equation}
When the columns of $\widehat{\mathcal{L}}^{ij}$ are mutually orthogonal, this expression is exact. However, if the columns of $\widehat{\mathcal{L}}^{ij}$  are not mutually orthogonal, then this will only be an approximation.  
\subsection{Expression for the Error in the Approximation.}
The simplest way to compute whether the columns of $\widehat{\mathcal{L}}^{ij}$ are mutually orthogonal is to compute their pairwise dot products. This is given compactly as $\widehat{\mathcal{L}}^{ij\top}\widehat{\mathcal{L}}^{ij}$, a symmetric matrix whose $i,j$th element is given by $\langle \widehat{v_i},\widehat{v_j}\rangle$.  The diagonal elements will always be 1, as the columns are normalized ($\langle \widehat{v_i},\widehat{v_i} \rangle=\|\widehat{v_i}\|^2=1$) and the off diagonal elements $\langle \widehat{v_i},\widehat{v_j} \rangle=0$ when columns are orthogonal and $\langle \widehat{v_i},\widehat{v_j} \rangle>0$ otherwise.  Naturally, if we subtract $\widehat{\mathcal{L}}^{ij\top}\widehat{\mathcal{L}}^{ij}$ from the identity matrix $\mathbf{I}$, then the diagonal elements will go to zero and when all of the columns are mutually orthogonal, the off-diagonal elements will be zero as well.  Thus the Frobenius norm of this matrix will be zero when the approximation is exact.  Thus we posit an expression for the error as follows,
\begin{equation*}
\Delta(\mathcal{L}^{ij}) = \|{\mathbf{I} -\widehat{\mathcal{L}}^{ij\top}\widehat{\mathcal{L}}^{ij}}\|_F
\label{eq:projection_error}
\end{equation*}
This expression is a natural measure for the degree to which the columns of $\mathcal{L}^{ij}$ are mutually orthogonal, thus we call it the \emph{orthogonality} of the matrix $\mathcal{L}^{ij}$ and denote it with the symbol $\Theta(\mathcal{L}^{ij})$ with the convention that $\Theta(\mathcal{L}^{ij}) = 1-\Delta(\mathcal{L}^{ij})$.
\begin{equation}
\Theta(\mathcal{L}^{ij}) = 1 - \|{\mathbf{I} -\widehat{\mathcal{L}}^{ij\top}\widehat{\mathcal{L}}^{ij}}\|_F
\label{eq:orthogonality}
\end{equation}
Finally, we prove that this expression, which quantifies the degree of orthogonality of the matrix $\mathcal{L}^{ij}$, is in fact a bound on the error in the approximation we introduced in Equation \ref{eq:projection_ratio_orth}.  Let us assume that the true ortho-normal is given as $\widehat{\mathcal{L}}_{\text{orth}}^{ij}$. 
\begin{eqnarray*}
\label{error_derivation}
\|{\rm proj}_{\mathcal{L}^{ij}}(v) &-& {\rm proj}_{\hat{\mathcal{L}}^{ij}_{\text{orth}}}(v)\| =...\\
&...& \| \hat{\mathcal{L}}^{ij}(\hat{\mathcal{L}}^{ij\top}\hat{\mathcal{L}}^{ij})^{-1}\hat{\mathcal{L}}^{ij\top}v - \hat{\mathcal{L}}^{ij}\hat{\mathcal{L}}^{ij\top}v \| \\
&\leq& \| \hat{\mathcal{L}}^{ij}(\hat{\mathcal{L}}^{ij\top}\hat{\mathcal{L}}^{ij})^{-1}\hat{\mathcal{L}}^{ij\top} - \hat{\mathcal{L}}^{ij}\hat{\mathcal{L}}^{ij\top} \| \|v\| \\
&=& \|{\mathbf{I} -\widehat{\mathcal{L}}^{ij\top}\widehat{\mathcal{L}}^{ij}}\|\|v\| \\
\end{eqnarray*}
where $A(A^{\top}A)^{-1}A^{\top}$ is a general projection matrix onto the column space of $A$, and $AA^{\top}$ is the projection matrix onto $A$ in the case that the columns of $A$ form an ortho-normal basis, which is easy to verify, as the term $(A^{\top}A)^{-1}=\mathbf{I}$.  The second line is given by Cauchy-Schwartz, and the third equality is proven in Appendix \ref{app:projection_error_proof}.  In general, orthogonality is a function of both the number of nodes in the network and of the rate constants.  In this work, we were mostly focused on comparing orthogonality between networks with the same number of nodes (and the same edges) when the rate parameters on those edges vary.  The error bound $(\Delta\mathcal{L}_{ij})$ that we calculate is actually a sort of ``non-orthogonality'', as it increases as the columns of the Laplacian become less orthogonal. There is a maximum ``non-orthogonality'' on a graph with a given number of nodes $(N)$ which is given by, $\| \mathbf{I}-\mathbf{1} \|_F$ where $\mathbf{I}$ is the identity matrix of size N and $\mathbf{1}$ is the matrix of all ones of size N.  This can be given in terms of $N$ as $\sqrt{N^2-N}$. Thus, the orthogonality for a graph of size $N$ can in principle fall in the range $(1, 1-\sqrt{N^2-N})$.  In graphs where the nodes and edges are fixed, these values will be even more constrained, as some of the entries of $\mathcal{L}_{ij}$ are forced to be 0 where edges are absent. Furthermore, the remaining non-zero entries must form a Laplacian matrix, with the diagonals set such that the column sums are 0.

\begin{figure*}[htb]
\includegraphics[width = 0.9\textwidth]{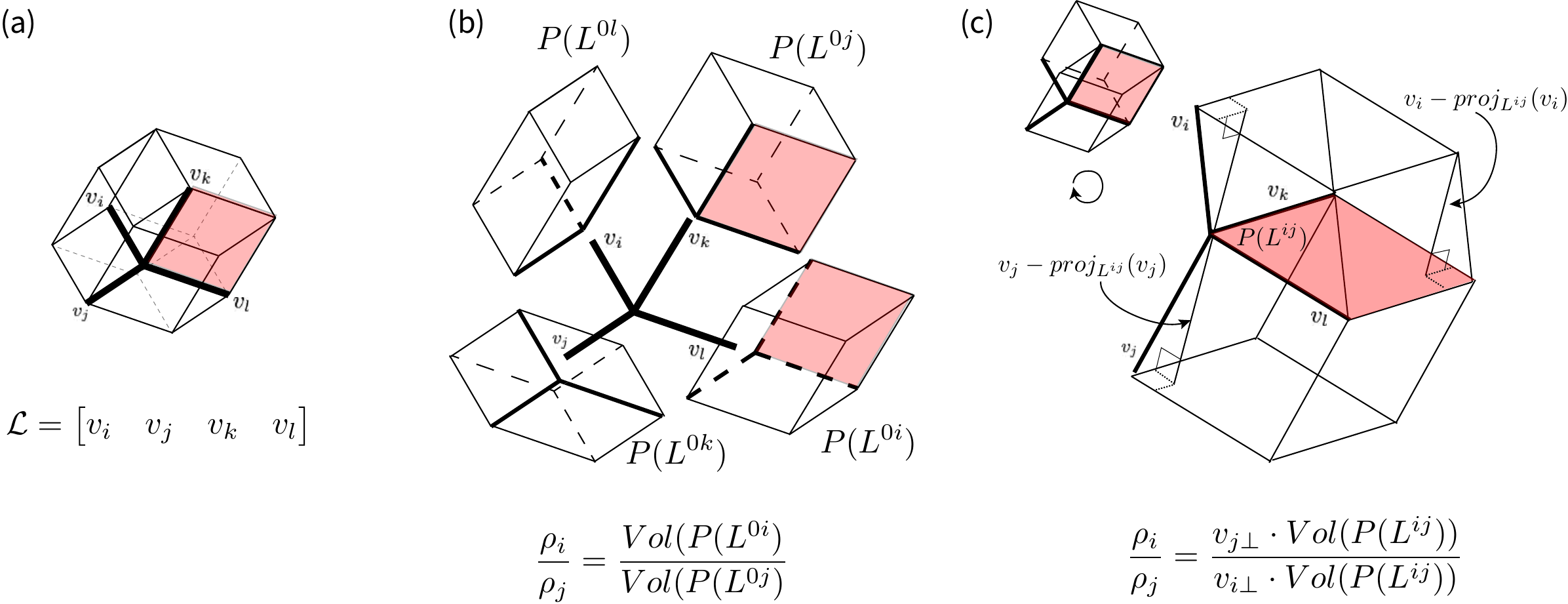}
\caption{Geometry of Laplacian Polytopes.  As an example, consider a 4-species chemical network given by the Laplacian Matrix $\mathcal{L}$ with 4 columns $(v_i,v_j,v_k,v_l)$. These columns can be thought of as 4-vectors that span a three dimensional subspace, and define 4 polytope volumes (a)  The polytope associated with species $i$ is the one formed from the columns of $\mathcal{L}$ remove $i$. An exploded view of these four polytopes is shown in (b). The ratio of any two steady state concentrations is equivalent to the ratio of the volumes of their corresponding polytopes.  These polytopes will share a facet, (for species $i$ and $j$, the shared facet is the one remove both columns $i$ and $j$, $\mathcal{L}^{ij}$ and shaded red).  Isolating only the polytopes associated with $i$ and $j$, and rotating them, we can see that the ratio of volumes can be expressed by the ratio of the product of this shared base and each `height' (c).  In this example the facet $\mathcal{L}^{ij}$ is two-dimensional, in general it will be an $(n-2)$-dimensional volume. However, this volume does not need to be calculated, as it appears in both the numerator and denominator of the ratio.  Thus the ratio of interest simplifies to the ratio of the heights $\|{v_j - {\rm proj}_{\mathcal{L}^{ij}}(v_j)}\|/\|{v_i - {\rm proj}_{\mathcal{L}^{ij}}(v_i)}\|$ 
\label{fig:3d_polytopes}}
\end{figure*}

\section{Results}
In the Preliminaries section, we presented an approximation for the discrimination ratio and an error bound for this approximation.  The error bound is given by the degree to which the column space of a subset of the generator matrix is mutually orthogonal, and thus we call this bound the \emph{orthogonality} of the matrix (Eq. \ref{eq:orthogonality}).  Here we will present two results related to orthogonality.  First, we show that orthogonality quantifies the degree to which a network is processive vs distributive, with processive networks in the low-orthogonality limit and distributive networks in the high-orthogonality limit. Second, we find that orthogonality is minimized in networks which discriminate based on binding energy differences (\emph{energetic discrimination}), and is maximized for networks which discriminate based on activation energy differences (\emph{kinetic discrimination}).  Taken together, these results show that processive networks are required for energetic discrimination, and distributive networks are required for kinetic discrimination.    

\subsection{Orthogonality in a Processive vs. a Distributive Network}
\begin{figure}[htb]
\includegraphics[width = 0.5\textwidth]{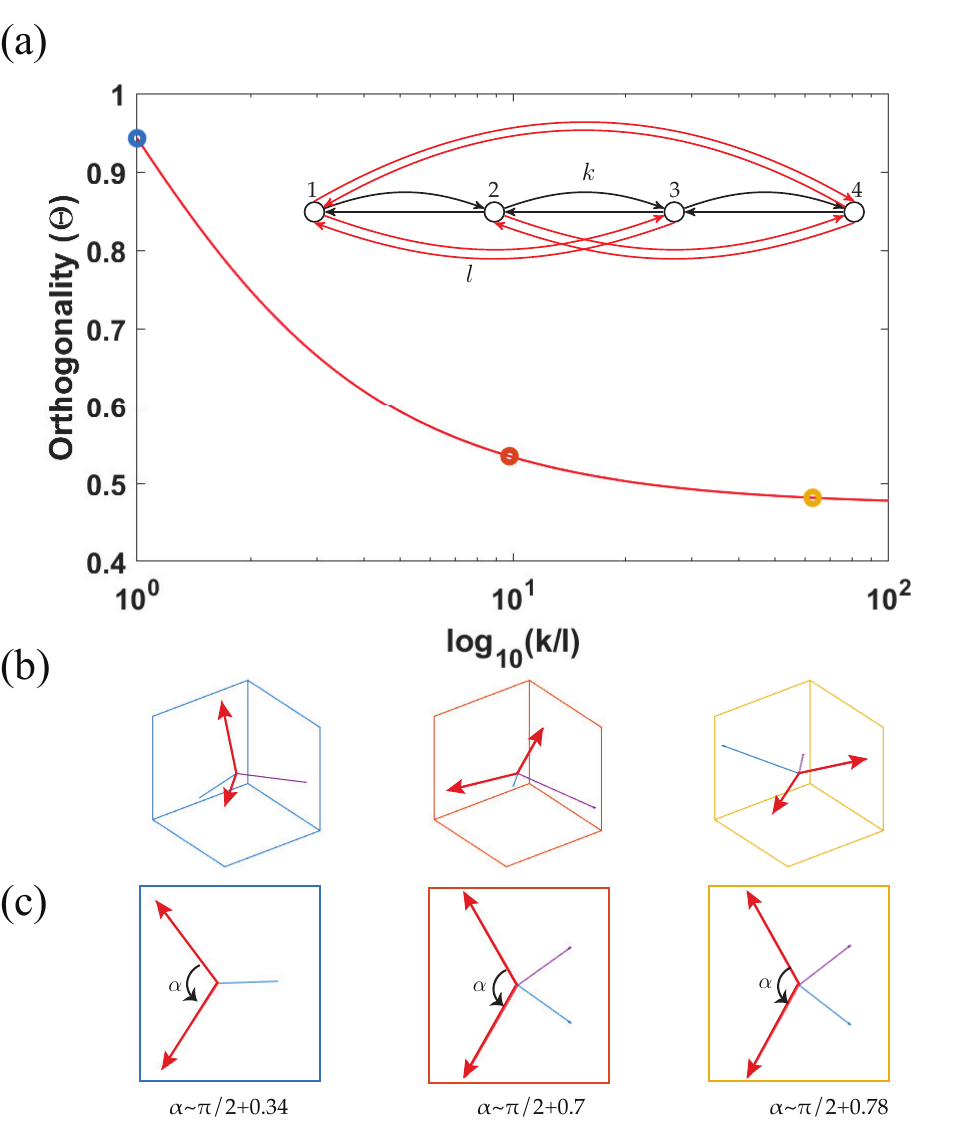}
\caption{Orthogonality captures the number of effective pathways directed at discriminatory nodes. Here we construct a simple example of a 4-species network in which we can tune the orthogonality.  If we construct a linear network with rates $k$, ((a) inset, black arrows), and $l$ ((a) inset, black arrows), the black pathways represents a single path between nodes species 1 and 4, while the red pathways represent "alternative pathways".  If we increase the ratio $(r=k/l)$ of rates of the black with respect to the red reactions $(r\gg1)$, a single pathway will dominate.  However, in the limit where all the rates are equal $(r=1)$ there are many effective pathways between 1 and 4.  Orthogonality gives a measure of these effective pathways in the network, where more orthogonal networks are more distributed.  This example shows directly the meaning of orthogonality.  The ratio of steady states can be computed using the projection onto the subspace spanned by the red vectors in (b) corresponding to the highlighted points (blue, orange, yellow) in (a), and becomes exact when these vectors form an orthogonal basis.  By rotating the vectors shown in (b), we can see that as the network becomes more processive by increasing $r$, these basis vectors become less orthogonal (c).
\label{fig:toy_model}}
\end{figure}
Here we introduce a simple 4-node toy model which demonstrates that orthogonality captures whether a network is distributive or processive. Consider a network [Fig. \ref{fig:toy_model}a, inset]  which has four nodes and in which the connections that would form a line graph (black arrows), are considered separately from the other connections (red arrows). If the reversible reactions represented by the black arrows have rate $k$, and those represented by the red arrows have rate $l$, we can, in this simple model, change the network from distributive to processive by changing the ratio $r = k/l$.  First consider the case when $r\gg1$ ($k\gg l$). In this case, there is a ``dominant path'' from the reactants (node 1) to the products (node 4), as reactions are much more likely to proceed along the black pathway as the rates along it are much faster than the pathways that use the red connections.  In this case we would say that the network is processive.  However, in the case when $r\approx1$ ($k\approx l$), this is not the case, the red reactions are just as fast as the black reactions, and this opens many equally good pathways between the reactants and products.  In this case, the network would be considered distributive.  Thus, for this simple toy model, as the ratio $r$ increases from $r=1$ to $r\gg1$, the network changes from distributive to processive.   If we look at the orthogonality of the network as we increase $r$, we see that it is decreasing [Fig \ref{fig:toy_model}a].  On advantage of this simple model is that the the orthogonality is calculated with respect to a 2-dimensional subspace spanned by two linearly independent vectors [Fig \ref{fig:toy_model}b, red arrows], shown for three different values of $r$.  In this case, the orthogonality is captured by a single value, i.e. the angle between these two vectors.  If we rotate the four vectors in Figure \ref{fig:toy_model}(b) so that we can visualize the ange ($\alpha$) between these two vectors, we see that as it approaches $\pi/2$ the orthogonality increases, as expected (Fig \ref{fig:toy_model}c).  These results suggest that in general, a network with line topology will have lower orthogonality than one with an all-to-all connected topology if all of the rates are of roughly equal magnitude.  We can compute that this holds in general for all N (Appendix \ref{app:all_vs_line}), where N is the number of nodes in the network. The increased orthogonality of the all-to-all relative to line topology captures a more general fact: orthogonality tends to decrease as connections are removed from a discrimination scheme, so long as these connections are of equal order magnitude to remaining connections, which we demonstrate computationally (Figure \ref{sfig:orth}).

\subsection{Orthgonality in the Hopfield-Ninio Discrimination Scheme}
\begin{figure*}[htb]
\includegraphics[width = 0.75\textwidth]{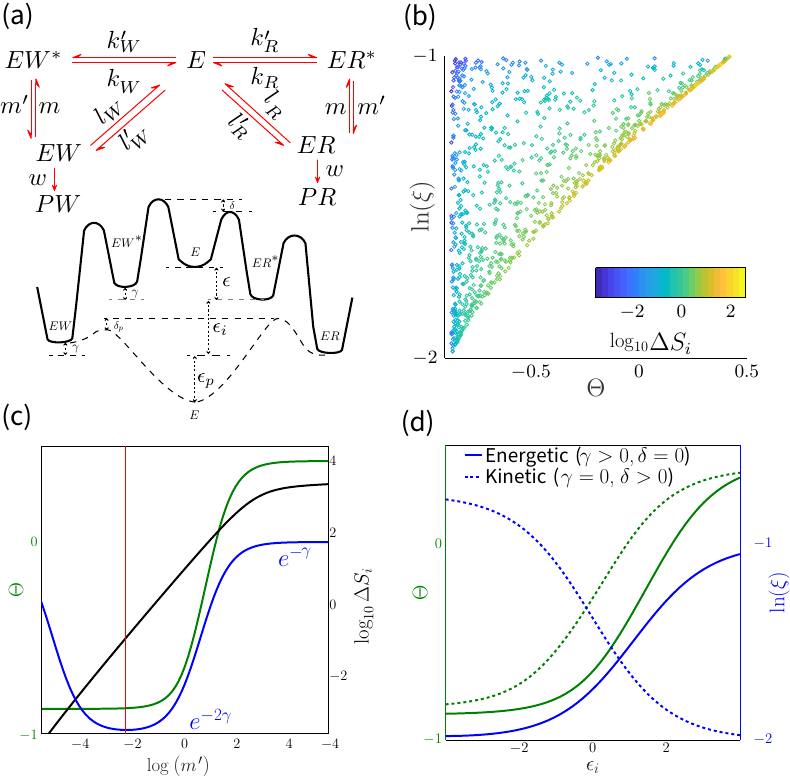}
\caption{Orthogonality in the Hopfield-Ninio scheme.  (a) Reaction diagram of the scheme with an associated free-energy diagram is shown (a, lower) where discrimination can occur due to the binding energy difference, $\gamma$, or the activation energy differences $\delta$ and $\delta_p$, between the R and W products in the first and proofreading reactions.  The second reaction $(m^{\prime}/m$) is identical for both substrates. For each cycle of this network, a total free energy of $\epsilon+\epsilon_i+\epsilon_p$ is consumed. (b) Orthogonality bounds minimum error rate in the energetic regime  ($\gamma=1, \delta,\delta_p=0$). The log of the error rate $(\log(\xi))$ as a function of the orthogonality ($\Theta$) is plotted for simulated data (parameter selection in Methods).   Heatmap coloration represents relative dissipation $\Delta S_i$ \cite{Schnakenberg1976};  for a given orthogonality, the error rate decreases as dissipation increases.  (c) In the energetic  regime, minimum error (red line, $\xi = e^{-2\gamma}$) is achieved by simultaneously minimizing orthogonality ($\Theta$, green) and maximizing dissipation (black).  Excess dissipation drives orthogonality upwards, approaching the binding energy difference ($\xi = e^{-\gamma})$ asymptotically.   (d) Orthogonality as a function of drive.  In the energetic regime (solid curves), error rate ($\xi$) is minimized in the limit of low orthogonality ($\Theta$). In the kinetic regime, (dashed curves), error rate is minimized in the limit of high orthogonality. For this scheme, the orthogonality is bounded by (1,-3.47). 
\label{fig:hopfield}}
\end{figure*}

We first demonstrate the relationship between orthogonality and discrimination in the classical Hopfield-Ninio scheme, shown graphically in Figure \ref{fig:hopfield}(a).  Here, substrates $S= \{W, \ R\}$ compete to form complexes with enzyme $E$.  `Wrong' and `Right' products are formed from substrates $W$ and $R$ (respectively), at rates proportional to the steady state occupancy of the final pre-catalysis complex $\rho_{ES}.$  We thus define the {\it error fraction} achieved by the discrimination scheme to be
\[
\xi = \frac{\rho_{EW}}{\rho_{ER}}.
\]
Ninio and Hopfield designed this scheme to amplify differences in the binding energies of $EW$ and $ER$ complex formation. Reaction rates are defined below in Equations \ref{eq:ew},\ref{eq:er}, and \ref{eq:m} following the Rao and Peliti \cite{Rao2015}. with the rate constants given in Kramer's form. A pseudo free energy diagram which corresponds to these definitions of the rate constants is shown in Figure \ref{fig:hopfield}(a, lower).  

We have for the $EW$ reactions:
\begin{eqnarray}
\label{eq:ew}
k^{\prime}_W = \omega e^\epsilon, &\hspace{4mm}&l^{\prime}_W = \omega_p \\
k_W =  \omega e^\gamma, & &l_W = \omega_p e^{\epsilon_p+\gamma}\nonumber
\end{eqnarray}
where: $\omega, \ \omega_p$ set overall rates; $(\epsilon-\gamma)$ is the enthalpy difference between $E$ and $EW^*$ and $(\epsilon_p+\gamma)$ is the free energy difference between $EW$ and $E$.
The $ER$ reactions are given by:
\begin{eqnarray}
\label{eq:er}
k^{\prime}_R = \omega e^{\epsilon+\delta},&\hspace{4mm}&l^{\prime}_R = \omega_pe^{-\delta_p} \\
k_R =  \omega e^\delta, & & l_R = \omega_p e^{\epsilon_p-\delta_p}\nonumber
\end{eqnarray}
For the `right' reactions, $\epsilon$ is the enthalpy difference between $E$ and $ER^*$ and $\epsilon_p$ is the difference between $ER$ and $E$. $\delta$ and $\delta_p$ set the activation energy differences between right and wrong complexes for the first and proofreading reaction respectively.

There is no discrimination along the transitions between the intermediary and pre-catalysis states:
\begin{equation}
\label{eq:m}
m = \omega_i, \hspace{4mm} m' = \omega_i e^{\eps_i}
\end{equation}

Note that for both the $R$ and $W$ reactions cycles, the total free energy consumed in a cycle from $E$ to $ES^*$ to $ES$ and back to $E$ is equal to $(\epsilon+\epsilon_i+\epsilon_p)$ in both cases (the $\gamma$ cancels for the $W$ side).  Thus, no consistent free-energies can be assigned to the states unless this sum is equal to zero and the system is in equilibrium.  However, we are free to choose the values of $\epsilon$, $\epsilon_i$, and $\epsilon_p$, and their sum will, in general, be non-zero.  

We begin by considering the relationship between error and orthogonality in the regime which is governed only by binding energy differences ($\gamma >0, \ \delta=0$), termed the `energetic regime'.  The Hopfield-Ninio scheme was originally designed for discrimination in this regime.  Simulations reveal that low orthogonality is necessary, but not sufficient, for low error rates in the energetic regime [Figure  \ref{fig:hopfield}(b)].  

In the original Hopfield scheme, it was already clear that enhanced discrimination beyond the equilibrium limit was only possible in certain parameter regimes.  In the following, we show how we can use orthogonality to find these regimes.  In schemes based on binding energy differences, orthogonality must be minimized and dissipation maximized for optimal discrimination.  Let us start by looking at the limit, long appreciated to be one of the limits required for the Hopfield-Ninio scheme to reach its lowest error, $\xi_{energetic}\to e^{-2\gamma}$.
\begin{equation}
\frac{\omega_p}{\omega e^{\epsilon}}\rightarrow0 
\label{eq:hopfiel_limit_1}
\end{equation}
Hopfield argued for the necessity of this limit (Eq. \ref{eq:hopfiel_limit_1}) by pointing out that if $\omega_p>\omega e^{\epsilon}$ then the reaction would favor simply bypassing the intermediate and forming the product directly.  Bypassing the intermediate state would destroy the enhanced discrimination.  We demonstrate that orthogonality is monotonically decreasing as this limit is approached (Appendix~\ref{app:hopfield}) which provides an alternative explanation as to why this limit is necessary.

A less well-appreciated requirement for energetic discrimination concerns the nonequilibrium drive, generated in this case by adjusting $\epsilon_i$ such that $\lvert(\epsilon+\epsilon_i+\epsilon_p)\rvert$ increases.  Some amount of drive is crucial for the discrimination scheme to be able to achieve error rates lower than the equilibrium free energy difference of the products $\gamma$, but too much drive will destroy this enhanced discrimination \cite{Wong2018-ys}.  We can understand this nonlinearity in terms of orthogonality (Figure  \ref{fig:hopfield}(c)).  Energy dissipation is helpful for discrimination up {\it until} the point at which it begins to drive up orthogonality.

We next turn to the regime governed by only activation energy differences ($\gamma =0, \ \delta > 0$), termed the `kinetic regime'.  Simulations reveal a bound opposite to that of the energetic regime: high orthogonality is necessary (but not sufficient) for low error (Supplemental Figure \ref{sfig:kin_sim}).  Analytically, we can derive the error in this regime to be
\begin{equation}
\xi_{\text{kinetic}} = \frac{1+e^{-\delta}\eta b+e^{-2\delta} \eta c}{1+\eta b+\eta c}
\end{equation}
where 
\[
a = \omega\omega_i, \ \ \ b = \omega\omega_p, \ \ \ c = \omega_p\omega_i e^{\epsilon_i}, \ \ \ \eta=e^{\epsilon_p}/a.
\]
The $\xi_{\text{kinetic}}$ is minimized when $\eta\gg1$ and $c\gg b.$  That is, when there exists high drive ($\omega_i e^{\epsilon_i}\gg\omega)$ and free enthalpy product differences ($\epsilon_p\gg0$).  We demonstrate that orthogonality is monotonically {\it increasing} as these limits are approached (Appendix~\ref{app:hopfield}).

Differences between the two discriminatory regimes are summarized in Figure  \ref{fig:hopfield}(d).  Increasing the dissipative drive ($\eps_i$) increases orthogonality, which allows for kinetic discrimination but precludes energetic discrimination.

\begin{figure}[htbp]
\includegraphics[width = 0.5\textwidth]{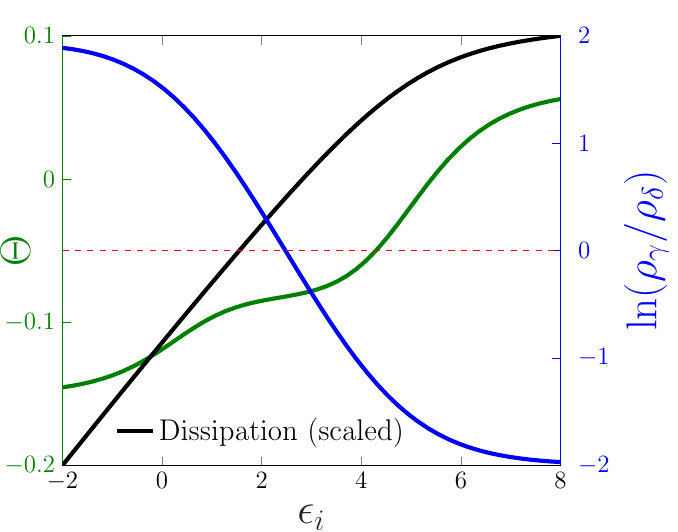}
\caption{A Hopfield-Ninio style network designed to tune product selectivity by modulating dissipation (black).  One product $\rho_{\gamma}$ has a lower binding energy and is favored in the energetic regime, while the other $\rho_{\delta}$ is has a lower activation energy and is favored in the kinetic regime.  The log of the ratio between the products ($\rho_{\gamma}$/$\rho_{\delta}$, blue), can be shifted from 2 ($\rho_{\gamma}$  favored) to -2 ($\rho_{\delta}$ favored) by driving across a single reaction.  This is due to orthogonality (green line) increasing, which shifts the network from the energetic to the kinetic regime. \label{fig:hopfield_switch}}
\end{figure}

The ability to modulate orthogonality via driving the second reaction via $\epsilon_i$ suggests a simple strategy for dissipation-driven product switching.
If products $EW, \ ER$ are favored by different energy types, they can be selected for by driving only the second reaction via $\epsilon_i$ such that the network moves from low to high orthogonality. We achieve a four order of magnitude selection effect via this scheme (Figure~\ref{fig:hopfield_switch}).
Because the Hopfield-Ninio scheme only has one intermediary product, it is difficult to interpret in terms of the number of effective pathways towards the discriminatory products.  In order to illustrate the connection between discrimination, effective pathways and orthogonality more clearly, we turn to a more general setting.

\subsection{Orthogonality in a General Setting}

\begin{figure*}[ht]
\includegraphics[width = 0.9\textwidth]{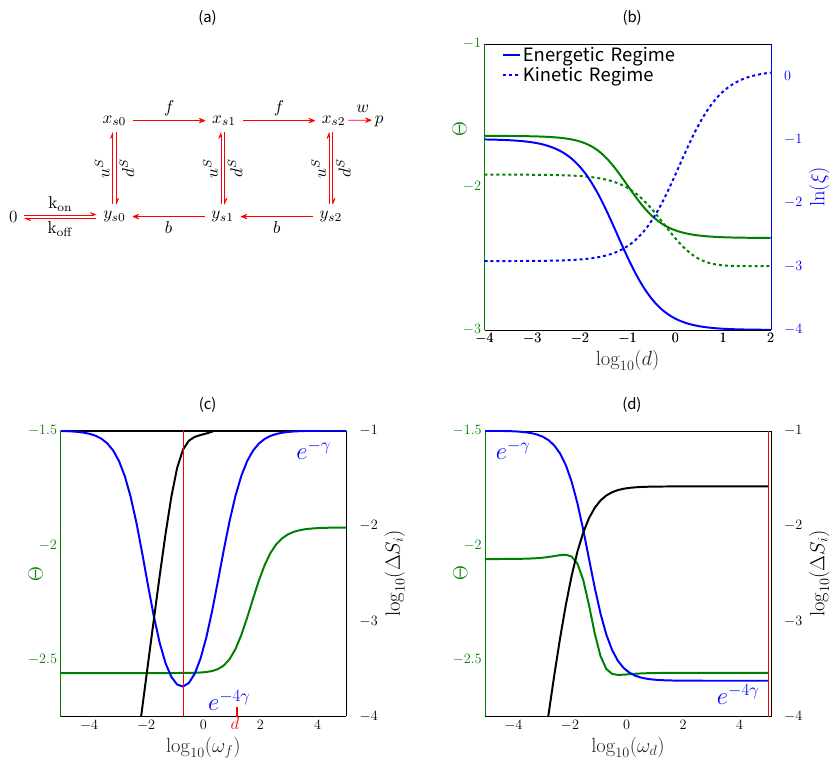}
\caption{(a) One side of the generalized ladder network \cite{Murugan2012}.  The full ladder contains a second side, symmetric about the $0$ node.  The two sides of the ladder have different $u^S, d^S$ constants ($S=\{R, W\}$ for `right' and `wrong' sides of the ladder, respectively). (b)  Orthogonality and error for the two-loop ladder.  In the energetic regime ($\delta$=0, solid curves), minimum error (blue) is achieved in the low orthogonality (green) limit.  In the kinetic regime ($\gamma$=0, dashed curves), minimum error is achieved in the high orthogonality limit.  (c) Non-monotonicity in the energetic regime.  The error rate ($\xi$, blue) is minimized (red line, $\xi=e^{-4\gamma}$ corresponding to $e^{-2\gamma}$ proofreading per loop) where dissipation (black) is maximized and orthogonality ($\Theta$, green) is minimized. Red tick indicates value of rate $d\approx15$. (d) Orthogonality is not always  an increasing function of dissipation.  Dissipation (black), error (blue), and orthogonality (green) for a two-loop ladder network in the energetic regime.  Note that the error rate is minimized (red line, $\xi=e^{-4\gamma}$) at lower dissipation than in the energetic-regime network at left (black line in (c) vs (d)) \label{fig:ladder}  In the ladder graph, the orthogonality is bounded by (1,-15.49).}.  
\end{figure*}

Murugan, Huse, and Leibler recently discovered that energetic discrimination in a general network requires a {\it discriminatory fence}~\cite{Murugan2014}, which can be idealized as a ladder graph having two sides, each with $N$ loops (Figure \ref{fig:ladder}(a)).  The sides of the ladder are symmetric about the $0$ node; the network aims to discriminate between states represented by its upper corners (i.e., $x_{s2}$ in Figure \ref{fig:ladder}(a)).  Rate constants $u^S,\ d^S, \ S=\{W, R\}$ will differ for the `Wrong' ($W$) and `Right' ($R$) sides of the network.

The ladder idealization captures the fact that a general energetic discrimination network must be processive and have a dominant `forward' ($f$) and `reverse' ($b$) path which are parallel to each other and effectively one-directional.  On the pathway towards the product state, there is the constant threat of `discard' ($d$), after which the reaction is exposed to a one-directional pathway away from the product state ($b$).  There is also the possibility of `rescue' ($u$) from discard. 

The Kramer's form rate constants for this network are
\begin{center}
\begin{tabular}{ll}
$u^R = \omega_d e^{\epsilon_u+\delta},$ & $d^R = \omega_d e^{\delta}$\\
\\
$u^W =  \omega_d e^{\epsilon_u}, $&$ d^W = \omega_d e^{\gamma}$.\\
\end{tabular}
\end{center}
And there is no discrimination ($f^R=f^W=f$) along the forward or reverse pathways:
\[
\begin{aligned}
f = \omega_f, & & b = \omega_b,\\
\end{aligned}
\]
which we approximate to be one-directional for analytical convenience, but treat as bidirectional with small reverse rates when necessary for computing dissipation.

It is clear from the Kramer's form constants that to discriminate in the energetic regime (i.e., via $\gamma$), a high discard rate ($d$) is required.  Indeed, the error rate for an $N$-loop network~\footnote{An $N$-loop network will strictly speaking be composed of $2N+1$ loops, $N$ on each side of the ladder and a single reactant node.} in this regime is 
\begin{equation}\label{energy_error}
\xi_{\text{energetic}} = \frac{1}{e^\gamma}\left (\frac{\omega_d+\omega_f}{\omega_d e^{\gamma}+\omega_f} \right )^N
\end{equation}
which achieves its minimum when discards are high relative to steps toward reaction completion:
\begin{equation}\label{elim}
\omega_d/\omega_f \to \infty.
\end{equation}
Discrimination in this regime is fundamentally processive, and global: accuracy relies on sequential exposure to frequently realized discard pathways, and {\it each} reaction step contributes to discrimination via the potential for discard.  Correspondingly, orthogonality monotonically decreases in the Equation \ref{elim} limit (Appendix \ref{app:ladder_error}), and is minimized in the high discard regime (Figure \ref{fig:ladder}(b), solid lines).

In contrast, we find that the kinetic regime has error fraction given by  (Appendix~\ref{app:ladder_error}):
\begin{equation}\label{general:kinetic_error}
\xi_{\rm kinetic} =  \frac{(\phi+1)^\alpha(1+\eta e^\delta)^\alpha}{(\phi e^\delta+1)^\alpha(\eta+1)^\alpha}.
\end{equation}
where 
\[
\phi = \omega_d e^{\epsilon_u}/\omega_b, \ \ \ and \ \ \ \eta =\omega_d/ \omega_f.
\]
The error $\xi_{\text{kinetic}}$ is minimized when $\eta\to0$ and $\phi\to\infty$, which is to say that:
\begin{equation}\label{kin_lad}
\omega_d/\omega_f \to 0, \ \ \omega_de^{\epsilon_u}/\omega_b \to \infty.
\end{equation}
These limits imply that network dynamics are being pushed quickly towards the final product nodes ($\omega_f,\ \epsilon_u$ large, $\omega_b$ small).  This makes local discrimination possible; and indeed orthogonality is monotonically increasing in the Equation \ref{kin_lad} limit (Appendix \ref{app:ladder_orth}).

Quick movement towards final product nodes is in opposition to high discard rates; we can thus summarize the difference between the energetic and kinetic regimes by observing their difference with respect to the discard rate ($d$, Figure \ref{fig:ladder}(b) x-axis), which reveal the expected orthogonality-error relationships in the two regimes.  Note that these limits correspond to the dynamical phase localization limits described in \cite{Murugan2016}.

We are now in a position to understand the orthogonality of this model in terms of its effective pathways towards the final product nodes.  The energetic discrimination requirement that $f<<d$ means that the network effectively contains only a single pathway to the product.  Intuitively, the single pathway results from the slowness of one-directional progress towards the final product; rescue pathways cannot add additional paths to the final product because they are effectively equilibrated relative to the slow forward progress.  Corresponding to this intuition, we find analytically that $u, \ b,$ have essentially no effect on orthogonality in the $f<<d$ regime (Appendix \ref{app:ladder_orth}).  This argument is consistent with the fact that the discrimination error in the energetic regime (Equation \ref{energy_error}) is independent of $u, \ b,$ but in the kinetic regime, which requires $d<<f,$ we find that $u, b$ are important factors in the error expression (Equation \ref{general:kinetic_error}) and orthogonality requirements (Equation \ref{kin_lad}).

In the energetic regime, we observe that as $f$ becomes close to $d$ (red tick, Figure \ref{fig:ladder}(c)), orthogonality rises sharply.  We understand this to result from many more effective pathways now leading to the final product.  Again, the rise in orthogonality as we increase $f$ leads to the non-monotonic behavior of the error rate.  

Our understanding of orthogonality in terms of effective pathways allows us to apply thermodynamic drive in the energetic regime such that drive does {\it not} increase orthogonality.  Our arguments above state that $f << d,$ enforces the single pathway and hence maintains orthogonality.  Therefore, if we dissipate energy to drive $d,$ we should find that the orthogonality decreases, and indeed we do [Figure \ref{fig:ladder}(d)].  Note that Figure \ref{fig:ladder}(c) was generated with the same parameters as Figure \ref{fig:ladder}(d); all that's changed is the reaction we choose to drive.  In this parametric limit, the orthogonality and dissipation requirements are not contravening.  

\begin{figure}[htbp]
\includegraphics[width = 0.5\textwidth]{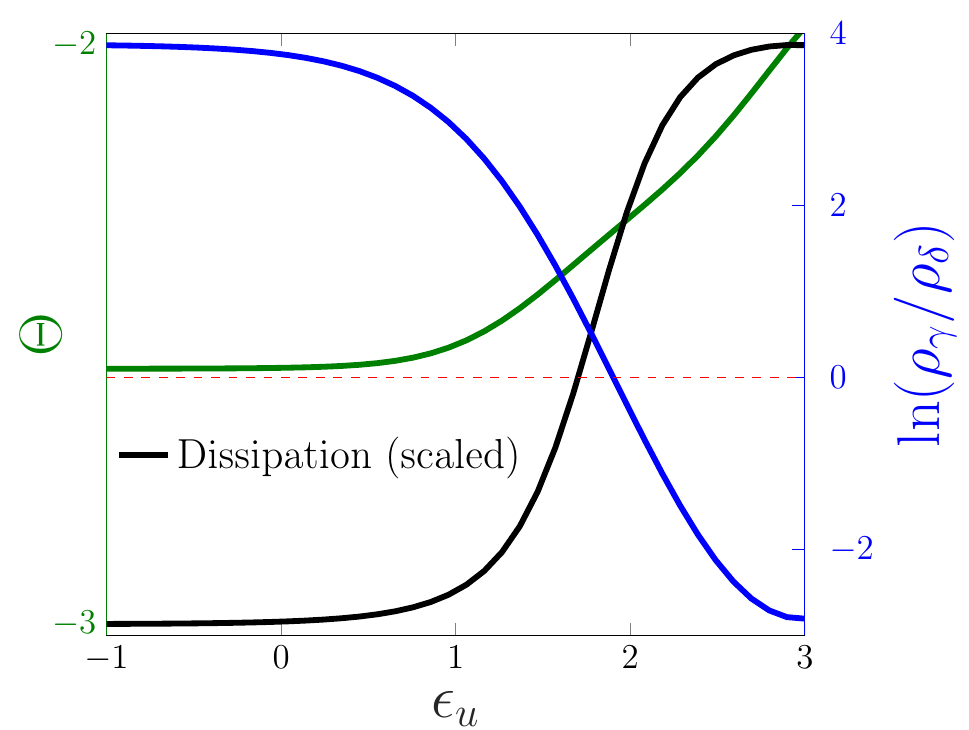}
\caption{The general ladder network can also achieve sensitive product switching.  In this network, binding energies favor the product ($\rho_\gamma$) on one side of the ladder while activation energies favor the other product ($\rho_\delta$).  Dissipation is used to drive $\epsilon_u,$ increasing the ratio of rescues to discards $u^S/d^S,$ thereby shifting the network from low orthogonality ($\rho_\gamma$ favored) to high orthogonality ($\rho_\delta$ favored). \label{fig:ladder_switch}}
\end{figure}
Finally, we note that (as in the Hopfield-Ninio regime) highly selective  - seven orders of magnitude - dissipation driven product switching is possible between states which are favored by different energy types (Figure \ref{fig:ladder_switch}).

\section{Discussion}
We have introduced a measure, which we call orthogonality, that was derived from an error bound on an approximation for the steady state ratio of states in a general non-equilibrium network which can be described by a master equation.
This of course presents some limitations, foremost, we require that the dynamics can be linearized, that is that they can be represented by a set of linear differential equations in the form $\frac{d\mathbf{p}}{dt} = \mathcal{L}\mathbf{p}$.  This does not limit the classes of reactions as much as it might at first seem, as many networks whose microscopic interactions are governed by non-linear differential equations may be linearized with carefully defined states and edge labels \cite{Gunawardena2012-eg} or by an appropriate coarse graining \cite{Costa2021-ij}.  For example, a linearization of the classic enzyme based catalysis scheme can be derived from the non-linear mass-action equation by including substrate concentration in an edge label.  Interestingly, this recovers the classic time-scale separation assumed to derive the  Michaelis-Menten equation \cite{Gunawardena2012-eg}.  

We propose that this orthogonality quantifies the degree to which such a network is processive versus distributive, and show that processive networks, which have a single dominant pathway between reactants and products, are characterized by low-orthogonality, while distributive networks which have many realizable paths, have high orthogonality.  In order to discriminate via binding energies, a processive network is required because discrimination is achieved by frequently discarding intermediates from the dominant path.  For such inherently processive processes, discrimination is a global function of discards at sequential steps throughout the graph.  Final product formation is rare, thus slow.   In contrast, discrimination via kinetic barriers is fast.  In the kinetic regime, discrimination relies on creating final products quickly, enabled by distributive networks which have many paths towards the final products. These results help to explain why "rescues" in general energetic discrimination schemes increase speed at the cost of accuracy, as increasing the rates of such reactions increases network orthogonality, which is beneficial for speed but detrimental to accuracy in energetic schemes.

Our results suggest that orthogonality is related to the degree of processivity or distributivity in a network, however, we do not have mathematical proof of this relationship.  This is in part because orthogonality is the only measure we know of which quantifies this aspect of networks, and thus we have nothing to compare it to directly.  While no other measures seem to capture the number of effective pathways in the same way, we can compare it to other graph theoretic measures, such as the graph sparsity and we do indeed find that orthogonality decreases as graphs become more sparsely connected (Figure \ref{sfig:orth}). It is interesting to note that activation energy differences are symmetric changes to the Laplacian, while binding energy differences are not, this may be significant to our understanding of why activation energy differences require high-orthogonality and binding energy differences require low-orthogonality.   It is also interesting to note that we can view this recursive orthogonalization procedure as the source of the extreme parametric complexity in general expressions for the discrimination ratio.  It is likely that for equilibrium systems, many symmetries simplify the orthogonalization and result in the simple expressions we are familiar from detailed balance, although it is beyond the scope of this work to derive those. 

It is interesting to consider this result in the context of protein complex assembly~\cite{Murugan2014a}.  Sartori and Leibler~\cite{Sartori2019} have recently proposed that a significant proportion of the discrimination necessary for accurate protein complex assembly can be achieved by equilibrium energy differences in protein-protein interactions.  Our results predict that non-equilibrium mechanisms which amplify these energetic differences should result in complexes being assembled sequentially, and slowly.  If non-equilibrium mechanisms instead amplify kinetic differences to achieve accurate assembly, we expect a complex's component subunits to assemble in many different orders, quickly.  

One potential use for this work is to provide a general procedure in which to find the parametric limits for a network which permit enhanced non-equilibrium discrimination.  The  parametric landscape for general networks is complex and it is difficult to optimize accuracy.  In networks with relatively few species, there regimes can be found intuitively, as was done for the Hopfield-Ninio scheme, but for larger networks, until now the only way to find the appropriate parameters is by brute force sampling.  This was the approach taken in both \cite{Murugan2014a} and \cite{Rao2015}.  However, our measure provide a principled way to perform a parameter search, by looking for parameter combinations that minimize orthogonality in energetic schemes, and maximize it in kinetic schemes.  This may be useful practically for modeling and simulation of discriminatory networks, and or optimization of networks using orthogonality as an easy to compute objective function, as computation of orthogonality should scale $O(n^2)$ while computation of the discrimination directly via SVD or matrix inversion would scale as $O(n^3)$, for example, a brute force search of 10,000 parameter combinations in the Hopfield scheme using Matlab on a 3.3 GhZ Intel i7 took about 1.54 seconds when computing discrimination using an SVD, while the computing orthogonality took only 0.040 seconds.  In some cases, analytical expressions for the orthogonality in certain parametric limits may also be tractable. 

Furthermore, our results clarify the role of thermodynamic drive in nonequilibrium discrimination.  We find that both kinetic and energetic discrimination are enhanced by increasing dissipation, but are subject to necessary requirements on orthogonality, which itself can be modulated upwards or downwards by free energy expenditure.  When dissipation and orthogonality requirements contravene one another, discrimination schemes will have error rates that are non-monotonically increasing with the dissipation. This not only explains the observation of such behavior for a well-known discrimination scheme, but also leads naturally to the idea of modulating orthogonality to select between energetically or kinetically favorable products.  We show that by modulating orthogonality with energy expenditure, discriminatory networks can indeed achieve sensitive product switching.  In particular, driving a {\it single} reaction type is sufficient for sharp selection between products, if the products are favored by different energy types and if the driving shifts the orthogonality of the network.

Networks which are capable of switching from processivity to distributivity may be ubiquitous in biochemical systems.  The ladder topology network shown in Fig \ref{fig:complex_formation}(c, d) is an abstraction and can be useful to describe many different cellular processes.  In general, the substrate need not be a protein and the modification need not be phosphorylation, this network could equally describe, e.g., a reaction complex forming around a nucleic acid substrate with methylation as the modification.  In fact, with a nucleic acid substrate, the modification could even be the nucleic acids's own self-association into a stem loop.  In this case, the ``removal'' of the modification could be driven by the activity of a helicase and modulated by ATP availability or by helicase gene expression for example. 

Biologically, this possibility may be realized in cytoplasmic ribonucleoprotein (RNP) granules~\cite{Brangwynne2009}.  These granules are composed of RNAs and proteins co-localized in liquid-liquid phase separated droplets.  Their components interact promiscuously and are known to be enriched for multivalent components~\cite{Banani2017}, which we propose may serve to increase distributivity and thus orthogonality.  RNA contributes to promiscuous granule interactions via both RNA-RNA interactions and serving as a protein scaffold ~\cite{Groot2019}. RNA structure is appealing as a modulator of orthogonality because it can be modified by driving a single reaction type.  It has been recently reported that ATP within granules is hydrolyzed by DEAD-box proteins, which remodel RNA by unwinding duplexes \cite{Hondele2019}.  This ATP-driven unwinding of RNA has been reported responsible for the dynamic makeup of RNA inside of granules, and for granule dissolution.  It is possible that driving this reaction type can tune the orthogonality of granule interaction networks, perhaps allowing for exploration of novel interactions among components. Such an ability is consistent with the apparent importance of granules in a wide variety of cellular responses to environmental cues, including stress response~\cite{Buchan2009}, transcriptional regulation~\cite{Anderson2009}, and local, activity dependent translation of mRNA at neuronal synapses~\cite{McCann2011, Barbee2006}.  From the theoretical side, it would be interesting to investigate how orthogonality changes in a physical model of phase separation.  Experimentally, it would be exciting to engineer a discriminatory network in which we can tune the orthogonality, and measure the resulting speed, accuracy, and product space directly.

\section*{Acknowledgements}
The authors would like to thank Tom Shimizu for useful discussions helping us to clarify the meaning of orthogonality and Gergo Bohner and Greg Wayne for useful discussions and Gergo Bohner and Pablo Sartori for critical reading of the manuscript.

\vspace{2mm}
The authors declare no competing or conflicting interests.

\vspace{2mm}
This work is partially supported by grants from the Wellcome Trust (104640/Z/14/Z, 092096/Z/10/Z) to E.A.M.  G.V. is supported by a grant from Emergent Ventures.  D.J. is funded by a Herchel Smith Post-doctoral Fellowship.  

\bibliography{references.bib}
\bibliographystyle{plain}

\appendix
\section{The discrimination ratio as a ratio of polytope volumes}
\label{app:ratio_volume_proof}

In this section we will prove that the ratio $\rho_i/\rho_j$ of the $i$th and $j$th elements of the steady state vector $\rho$ can be expressed as the ratio of the volumes of the polytopes associated with $i$ and $j$.  

The result follows from these equalities:
\begin{equation}\label{three_eq}
\begin{aligned}
\frac{\rho_i}{\rho_j} &= \frac{\det(\mathcal{L}_k^{j})}{\det(\mathcal{L}_k^{i})}  \ \ \ \forall \ k \in 1\ldots N\\
& = \frac{{vol}(P(\mathcal{L}^{0j}))}{{vol}(P(\mathcal{L}^{0i}))} \\
& =  \frac{\norm{v_i - \proj_\sij(v_i)}}{\norm{v_j - \proj_\sij(v_j)}}
\end{aligned}
\end{equation}
where $\mathcal{L}_k^{i}$ represents the matrix formed by removing row $k$ and column $i$ from matrix $\mathcal{L},$ and $\mathcal{L}^{0i}$ is formed from $\mathcal{L}$ by removing column $i$ only.  For matrix $A$, $Vol(P(A))$ represents the volume the parallelotope formed by the columns of $A$ and vector $v_i$ represents the $i$th column of $\mathcal{L};$ 

We proceed by proving each of the equalities.  To prove the first equality, it will be useful to have the definition of the adjugate matrix at hand.
\begin{defn}[Adjugate matrix]{
The components of the adjugate of a matrix $A$, $\adj(A),$ are given by taking the transpose of the cofactor matrix, $C$, of $A$:
\begin{equation}\label{adj}
\begin{aligned}
\adj(A)_{ij} &= C_{ji}\\
&= \det(A^{ji})\\
\end{aligned}
\end{equation}
where $A^{ji}$ is denotes the $(n-1) \times (n-1)$ matrix resulting form removing row $j$ and column $i$ from $A.$
}
\end{defn}

\begin{prop}[Discrimination ratio in terms of determinants with column and row cuts]{ We aim to demonstrate that
\[
\frac{\rho_i}{\rho_j} = \frac{\det(\mathcal{L}_k^{j})}{\det(\mathcal{L}_k^{i})}  \ \ \ \forall \ k \in 1\ldots N.
\]
}
\begin{proof}
The proposition was proved in \cite{Mirzaev2013}.  We include the argument here for completeness.  By the Matrix-Tree theorem, the rank of a strongly-connected $N$ dimensional Laplacian matrix is $N-1.$  The nullspace is therefore one-dimensional, and can be represented by a single basis vector $\rho$.  

It will suffice to prove that $\rho_i = \det(\mathcal{L}_k^{i})$. Recall the Laplace expansion for the determinant:
\begin{equation}\label{expand}
\begin{aligned}
\adj(\mathcal{L})\cdot \mathcal{L} = \mathcal{L}\cdot \adj(\mathcal{L}) &= \det(\mathcal{L})\cdot I\\ 
&= 0_{n\times n},
\end{aligned}
\end{equation}
where $0_{n\times n}$ denotes the $n$ by $n$ zero matrix and the final equality follows from $\mathcal{L}$ not being full rank, hence $\det(\mathcal{L}) = 0.$

Consider that $\mathcal{L}\cdot \adj(\mathcal{L})=0$ implies that $\mathcal{L} v=0_{n\times 1}$ for all $v$ which are columns of $\adj(\mathcal{L}).$  That is: the columns of $\adj(\mathcal{L})$ are equal to $\rho.$  This gives the result.
\end{proof}
\end{prop}

We now prove the second equality.

\begin{prop}[Discrimination ratio in terms of column cuts only]{We now wish to demonstrate that the equality presented in the previous proposition does not require the removal of some row $k$ \cite{Drucker2015}:
\[
\begin{aligned}
\frac{\det(\mL_k^{i})}{\det(\mL_k^{j})} &= \frac{Vol(P({\mL_k^{i}})}{Vol(P({\mL_k^{j}})}\\
&= \frac{Vol(P({\mL^{0i}})}{Vol(P({\mL^{0j}})}\\
\end{aligned}
\]
}
\begin{proof}
The first equality is a common characterization of the determinant.  The second result follows from a series of equalities 
\[
\begin{aligned}
\frac{vol(P(\mL^{0i}))}{vol(P(\mL^{0j}))} &= \frac{\sqrt{\det[(\mL^{0i})^T(\mL^{0i})]}}{\sqrt{\det[(\mL^{0j})^T(\mL^{0j})]}}\\
& = \sqrt{\frac{\sum_k (\det[\mL_k^{i})]^2} {\sum_k (\det[\mL_k^{j})]^2} }\\
& = \sqrt{\frac{N(\det(\mL_k^{i}))^2} {N(\det(\mL_k^{j}))^2} } = \frac{\det(\mL_k^{i})}{\det(\mL_k^{j})}\\
\end{aligned}
\]
where: the first equality is by definition of a polytope volume generated by a non-square matrix; the second equality results from applying the Cauchy-Binet formula; the third equality follows from noting that $\det(\mL_k^{i}) = \det(\mL_{k'}^{i}), \ \forall \ k,k' \in 1\ldots N.$
\end{proof}
\end{prop}

We now prove the final equality in Equation \ref{three_eq}.  First, it is useful to recall the base-height formula for determinants.

\begin{fact}[The base-height formula]{The determinant of a matrix $A$ can be written as
\[
\det(A) = \prod_i \lVert a_i \rVert
\]
where ${a_i}$ is a vector representing the component of $v_i$ that is perpendicular to the subspace spanned by the $N-i$ vectors \{$v_{i+1},\cdots,v_n$\}.  Crucially, this procedure can be done by selecting the $v_i$ in any order \cite{Gover2010}.
\begin{proof}
Geometrically, the determinant of a matrix $A$ having columns $v_i$ can be thought of as the volume of the  parallelotope generated by the $v_i.$  Consider a parallelotope $P(A)$ generated by vectors $\{v_{1},\cdots,v_n\}.$  $P(A)$ can also be thought of as a prism with base generated by the vectors $\{v_{2},\cdots,v_n\}$ and height equal to the magnitude of the component of $v_1$ perpendicular to the span of $\{v_{2},\cdots,v_n\}$.  It follows that
\begin{eqnarray*}
Vol_n(P(A)) = &Vol_{n-1}(P(\{v_{2},\cdots,v_n\})) \cdot \\ &\lVert v_1 - \proj(v_1; v_{2},\cdots,v_n)  \rVert
\end{eqnarray*}
And of course we can carry out this procedure successively for $Vol_{n-1}, Vol_{n-2},\ldots$.  This gives the desired result.
\end{proof}
}
\end{fact}

\begin{prop}[Discriminatory ratio in terms of normalized projections]{Finally, we demonstrate that
\[
\frac{Vol(P({\mL^{0i}}))}{Vol(P({\mL^{0j}}))} = \frac{\norm{v_j - \proj_S(v_j)}}{\norm{v_i - \proj_S(v_i)}}
\]
}
\begin{proof}
The result follows directly from the base-height formula for determinants.
\[
\begin{aligned}
\frac{\det(\mL^{0i})}{\det(\mL^{0j})} &= \frac{\norm{v_j - \proj_\sij(v_j)}\cdot Vol_{n-2}P(\{v_l\}_{l\neq i,j})}{\norm{v_i - \proj_\sij(v_i)}\cdot Vol_{n-2}P(\{v_l\}_{l\neq i,j})}\\
 &= \frac{\norm{v_j - \proj_\sij(v_j)}}{\norm{v_i - \proj_\sij(v_i)}}
\end{aligned}
\]
where $\proj_\sij(v_j)$ denotes the projection of vector $v_j$ onto the subspace spanned by the vectors of matrix $\sij,$ formed by deleting columns $i,\ j$ from $\mL.$
Notice that in the numerator, we have chosen to begin the base-height iteration with vector $v_j.$  Because $L^{0i}$ already has column $i$ removed, this procedure yields - in the numerator - a polytope base generated by the non-$i,j$ columns in $\mL.$  In the denominator, beginning the base-height iteration $v_i$ also yields a polytope base generated by the non-$i,j$ columns.  These bases cancel to give the desired result.
\end{proof}
\end{prop}

\section{Orthogonality is equivalent to the projection approximation error}
\label{app:projection_error_proof}
In this section, we aim to prove the following proposition.  
\begin{prop}[Projection approximation]{Let $S$ be a matrix having full rank (note that our $\mathcal{L}^{ij}$ are of full rank). We have that
\[
\begin{aligned}
\|{(S(S^\top S)^{-1}S^\top) - SS^\top}\| &= \|{I - S^\top S}\|.
\end{aligned}
\]
}
\end{prop}
\begin{proof}
Let $S$ have singular value decomposition $S=U\Sigma W^\top.$ 
\[
I - S^\top S = I - W\Sigma^\top\Sigma W^\top = W[I - \Sigma^\top\Sigma]W^\top.
\]
And similarly (noting that $S^\top S$ is invertible because $S$ is full rank):
\begin{widetext}
\[
\begin{aligned}
SS^\top - S(S^\top S)^{-1}S^\top &=U\Sigma\Sigma^\top U^\top - U\Sigma W^\top(W\Sigma^\top\Sigma W^\top)^{-1}W\Sigma^\top U^\top\\
&=U\Sigma\Sigma^\top U^\top - U\Sigma W^\top W(\Sigma^\top\Sigma)^{-1}W^\top W\Sigma^\top U^\top\\
&=U\Sigma\Sigma^\top U^\top - U\Sigma(\Sigma^\top\Sigma)^{-1}\Sigma^\top U^\top\\
& = U(\Sigma\Sigma^\top - \Sigma(\Sigma^\top\Sigma)^{-1}\Sigma^\top)U^\top\\
\end{aligned}
\]
\end{widetext}
It follows by direct calculation ($\Sigma$ is diagonal) that
\[
\|{\Sigma\Sigma^\top - \Sigma(\Sigma^\top\Sigma)^{-1}\Sigma^\top}\| = \|{I - \Sigma^\top\Sigma}\|.
\]
This gives the result.
\end{proof}

\section{Orthogonality of an equal weighted all-to-all graph is greater than that of a line graph}
\label{app:all_vs_line}
In this Appendix we demonstrate that the orthogonality of an $N$ node line graph is strictly less than an $N$ node all-to-all connected graph, in the toy case where all rate constants are the same.  The result follows from directly calculating the orthogonality for each topology, which we do in turn.

\begin{prop}[$\Theta$ for a line graph]{For a a line graph with bidrectional connections of equal weight (set to 1 without loss of generality), the orthogonality is given by: $\Theta  = 1 - \sqrt{(N-1)\frac{8}{9}+\frac{1}{36}(N-4)}$. 
}
\end{prop}
\begin{proof}
The result follows from direct computation of $\langle i, \ j\rangle, \ \forall \  i\neq 1,N.$
There are only two types of nonzero $\langle i,j \rangle.$  The first type is $\langle i, i+1 \rangle;$ there exist $2(N-1)$ terms of this type.  The second type is $\langle i,i+2 \rangle;$ there exist $N-4$ entries of this type.  The first type of nonzero term represents `neighbors.'  The second represents nodes separated by one node, which point at a mutual node.  The two types of inner product have (squared, normalized) values:
\[
\langle i, i+1 \rangle^2 = \frac{(-\alpha\cdot2\alpha-\alpha\cdot 2\alpha)^2}{(2\alpha^2+4\alpha^2)^2} = \frac{4}{9}
\]
and
\[
\langle i, i+2 \rangle^2 = \frac{(\alpha^2)^2}{(6\alpha^2)^2} = \frac{1}{36}.
\]
The result follows.
\end{proof}

The all-to-all calculation is slightly more complicated.

\begin{prop}[$\Theta$ for an all-to-all graph]{For an all-to-all connected graph with bidrectional connections of equal weight (set to 1 without loss of generality), the orthogonality is given by: $\Theta  = 1 - \sqrt{ \frac{(N-2)(N-3)}{(N-1)^2}}$. 
}
\end{prop}
\begin{proof}
Let $S$ be the $n$ by $n-2$ matrix formed by removing two of the columns of the Laplacian for this graph.

Because the diagonal elements $(S^\top S)_{ii} = 1,$ we need only compute the off-diagonal elements of $S^TS.$  A generic such element resulting from taking the (not normalized) inner product of columns $j, k$ is given by
\[
\begin{aligned}
\langle j, k \rangle &= \sum_{\tristack{i}{i \neq j}{i\neq k}}\theta_{ij}\theta_{ik} - \theta_{jk}\cdot\sum_{\sumstack{i}{i\neq j}} \theta_{ij} - \theta_{kj}\sum_{\sumstack{i}{i\neq k}} \theta_{ik}\\
& = (N-2)\alpha^2-\alpha^2(N-1)-\alpha^2(N-1)\\
&=-\alpha^2N.\\
\end{aligned}
\]
where the first line is a generic expression for the inner product of columns corresponding to connected nodes for matrix elements $\theta_{ij}$ of $S$, and the resulting lines follow from bidirectional all-to-all connectivity with equal rate constants.

We now need to compute the normalization factor:
\[
\begin{aligned}
\left(\|{j}\|\|{k}\|\right)^2 =& 
\left(\sum_{\eyej}\theta^2_{ij} + \left(\sum_{\eyej}\theta_{ij}\right)^2\right)\\
&\cdot \left(\sum_{\eyek}\theta^2_{ik} + \left(\sum_{\eyek}\theta_{ik}\right)^2\right)\\
&= \left(\alpha^2(N-1) + (N-1)^2\alpha^2\right)^2\\
&= \left(\alpha^2(N^2-N)\right)^2\\
& = \alpha^4(N^2-N)^2
\end{aligned}
\]
where again we have begun with generic terms for the normalization of the inner product of columns of the Laplacian matrix, with  $\theta_{ij}$ representing the elements of $S$.

Putting these together yields the expression for a generic element of $S^TS$:
\[
\begin{aligned}
\frac{\langle j, k \rangle^2}{\left(\|{i}\|\|{j}\|\right)^2} &= \frac{\alpha^4N^2}{\alpha^4(N^2-N)^2}\\
& = \frac{1}{(N-1)^2}.
\end{aligned}
\]
How many such elements exist?  We know that $S^TS$ is a square $n-2$ length matrix, and we know that the diagonal terms are zero.  We therefore have $(n-2)(n-3)$ entries each equal to $\frac{1}{(N-1)^2}.$  The result follows.
\end{proof}

From the two propositions we can calculate that
\[
\begin{aligned}
\Theta_{\text{all-to-all}} - \Theta_{\text{line}} =&
-\sqrt{ \frac{(N-2)(N-3)}{(N-1)^2}} \\
&+\sqrt{(N-1)\frac{8}{9}+\frac{1}{36}(N-4)}
\end{aligned}
\]
The former (negative) term quickly approaches 1, whereas the latter (positive) term grows as $O(\sqrt{N}).$  We conclude that the orthogonality of the all-to-all graph is greater than the line graph, and this difference is increasing for increasing $N$.

\section{Analytic expression for orthogonality in the 4-Node toy model}
\label{app:4node}

We will show how orthogonality changes as the graph in Figure \ref{fig:toy_model}(a) is modified, in support of the claims made in the main text.
Because we are discriminating between the end nodes, the orthogonality of the scheme in Figure \ref{fig:toy_model}(a) is a function of a single (normalized) inner product:
\begin{eqnarray}\label{node:inner}
\langle v_2,v_3\rangle^2 &=\frac{(2k(2k+l)-2kl)^2}{(2k^2+l^2+(2k+l)^2)^2} \nonumber\\
&=\frac{4 k^4}{\left(3 k^2+2 k l+l^2\right)^2}
\end{eqnarray}
with $k, l$ corresponding to black, red arrows in Figure \ref{fig:toy_model}(a), as defined in the main text.

We will first demonstrate how orthogonality changes as $r=k/l$ grows.  We then demonstrate how orthogonality changes upon removing the black (bidirectional) connection between the middle nodes.

\paragraph{Adjusting rates to favor a single path reduces orthogonality}
We can rewrite Equation \ref{node:inner} in terms of $r=k/l:$
\[
\langle v_2,v_3\rangle^2 =\frac{4 r^4}{\left(3 r^2+2 r+1\right)^2}.
\]
Two such terms contribute to the orthogonality giving
\begin{eqnarray*}
\Theta = &1-\sqrt{2 \langle v_2,v_3\rangle^2} \\
=&1-\sqrt{\frac{8 r^4}{\left(3 r^2+2 r+1\right)^2}}
\end{eqnarray*}
which is decreasing with $r$ as $O(r^{-2}),$ as claimed in the main text.

\paragraph{Removing a link}
What happens to the orthogonality when we remove the black bidirectional links between the middle nodes?

The expression for $\langle v_2,v_3\rangle^2_{\text{removed}}$ is given by
\[
\begin{aligned}
\langle v_2,v_3\rangle^2_{removed} &=\frac{k^2 l^2}{\left(k^2+k l+l^2\right)^2}\\
&= \frac{r^2}{\left(r^2+r+1\right)^2}
\end{aligned}
\] 
When $r\approx1$ this expression is equal to Equation \ref{node:inner}; there is no affect on orthogonality. However, as $r$ increases, $\langle v_2,v_3\rangle^2_{\text{removed}}$ becomes smaller than Equation \ref{node:inner}; deleting the connections increases orthogonality. 
We conclude that when $r>1$, the black bidirectional links form part of the dominant path, removing them will therefore increase the orthogonality.  

\section{Error and Orthogonality in Ninio-Hopfield Model}
\label{app:hopfield}
We first consider the Hopfield model in the energetic regime.  The Laplacian for this scheme with the columns corresponding to final products removed is given by
\[ 
A = \left(\begin{array}{ccc}
-\sum_1 & \oom e^\gam & \oom \\
\oom e^\eps & -\sum_2 & 0 \\
\oom e^\eps & 0 & -\sum_3 \\
\wpp & m' & 0 \\
\wpp & 0 & m' \\
\end{array}\right).
\]

Orthogonality in this model will be a function of three inner products:
\[
\Theta = 1-\sqrt{2*(s^2_{1,2} +s^2_{1,3} +s^2_{2,3} )}
\]
where we have denoted the (normalized) inner product of the $i$th and $j$th elements of $A$ as $s_{i,j}.$  It will be useful to define and reason about
\[
\sum{s^2_{i,j}}= (s^2_{1,2} +s^2_{1,3} +s^2_{2,3} ).
\]

The relevant inner products are
\begin{widetext}
\begin{eqnarray*}
s^2_{1,2} = \frac{\la 1,2 \ra^2}{\left(\norm{1}\norm{2}\right)^2}&=&\frac{(3e^\eps \oom^2 + 2\oom \wpp + e^\eps \oom m' - \wpp m')^2}
{4(3e^{2\eps}\oom^2+ 4 e^\eps\wpp \oom + 3 \wpp)(\oom^2 + \oom m' + m'^2)}\\
s^2_{1,3} = \frac{\la 1,3 \ra^2}{\left(\norm{1}\norm{3}\right)^2}&=&\frac{(3e^{\eps+\gamma} \oom^2 + 2\oom \wpp e^{\gamma} + e^\eps \oom m' - \wpp m')^2}{4(3e^{2\eps}\oom^2+ 4 e^\eps\wpp \oom + 3 \wpp)(\oom^2e^{2\gamma} + \oom m' e^{\gamma}+ m'^2)}\\
s^2_{2,3} = \frac{\la 2,3\ra^2}{\left(\norm{2}\norm{3}\right)^2}&=&\frac{\oom^4 e^{2\gam}}{4(\oom^2+\oom m' + m'^2)(e^{2\gam}\oom^2+e^\gam \oom m' + m'^2)}.\\
\end{eqnarray*}
\end{widetext}

We now demonstrate the orthogonality-discrimination relations made in the main text.   To do so, we first compute the orthogonality in the high and low discrimination limits, in order to demonstrate that orthogonality is lower ($\sum{s^2_{i,j}}$ higher) as high discrimination improves.  We will then compute the degree to which orthogonality movement between the low and high discrimination limits is monotonic.

In the energetic regime, the discrimination is maximized in the limits
\[\frac{\omega_p}{\omega e^{\epsilon}}\to0\].
We must therefore consider: $\omega\to\infty$,$\epsilon\to\infty$, and $\omega_p\to0.$

\paragraph{Energetic Limit 1: $\omega\to\infty$}
%In the energetic regime, we want to prove $\Theta$ decreases as accuracy increases.  For $\omega$, accuracy increases as $\omega$ goes from $0\to\infty$. Thus if, $\Theta$ decreases as  $\omega\to\infty$ we need to prove that $\sum{s^2_{i,j}}$ is increasing with increasing $\omega$ and that this increase is monotonic (note we have replace $m'$ with $\mu$).

Note that we replace $m'$ with $\mu$ in the below.

\subparagraph{$\sum{s^2_{i,j}}$ is increasing with $\omega$} To demonstrate this, we will show the following.  
\[
\lim_{\omega\to\infty}\sum{s^2_{i,j}}>\lim_{\omega\to0}\sum{s^2_{i,j}}
\]
Analytically, we can see that the in the limit of $\omega\to\infty$, only terms of order $\omega^4$ remain.  If we expand and collect the terms together in $\omega$
\begin{widetext}
\begin{eqnarray*}
s^2_{1,2}&=\frac{9 w^4 e^{2 \epsilon }+w^3 \left(12 \omega_p e^{\epsilon }+6 \mu  e^{2 \epsilon }\right)+w^2 \left(4 \omega_p^2-2 \mu  \omega_p e^{\epsilon }+\mu ^2 e^{2 \epsilon }\right)+w \left(-4 \mu  \omega_p^2-2 \mu ^2 \omega_p e^{\epsilon }\right)+\mu ^2 \omega_p^2}{12 w^4 e^{2 \epsilon }+w^3 \left(16 \omega_p e^{\epsilon }+12 \mu  e^{2 \epsilon }\right)+w^2 \left(12 \omega_p^2+16 \mu  \omega_p e^{\epsilon }+12 \mu ^2 e^{2 \epsilon }\right)+w \left(12 \mu  \omega_p^2+16 \mu ^2 \omega_p e^{\epsilon }\right)+12 \mu ^2 \omega_p^2}\\
s^2_{1,3}&=\frac{9 w^4 e^{2 \gamma +2 \epsilon }+w^3 \left(6 \mu  e^{\gamma +2 \epsilon }+12 e^{\gamma +\epsilon }\omega_p e^{\gamma }\right)+w^2 \left(-6 \mu  \omega_p e^{\gamma +\epsilon }+4\omega_p e^{2 \gamma }+4 \mu  e^{\epsilon }\omega_p e^{\gamma }+\mu ^2 e^{2 \epsilon }\right)+w \left(-4 \mu  \omega_p\omega_p e^{\gamma }-2 \mu ^2 \omega_p e^{\epsilon }\right)+\mu ^2 \omega_p^2}{12 w^4 e^{2 \gamma +2 \epsilon }+w^3 \left(12 e^{2 \epsilon } \text{$\mu $e}^{\gamma }+16 \omega_p e^{2 \gamma +\epsilon }\right)+w^2 \left(12 e^{2 \gamma } \omega_p^2+16 \omega_p e^{\epsilon } \text{$\mu $e}^{\gamma }+12 \mu ^2 e^{2 \epsilon }\right)+w \left(12 \omega_p^2 \text{$\mu $e}^{\gamma }+16 \mu ^2 \omega_p e^{\epsilon }\right)+12 \mu ^2 \omega_p^2}\\
s^2_{2,3}&=\frac{\oom^4 e^{2\gam}}{4 e^{2 \gamma } w^4+w^3 \left(4 e^{\gamma } \mu +4 e^{2 \gamma } \mu \right)+w^2 \left(4 e^{\gamma } \mu ^2+4 e^{2 \gamma } \mu ^2+4 \mu ^2\right)+w \left(4 e^{\gamma } \mu ^3+4 \mu ^3\right)+4 \mu ^4}
\end{eqnarray*}
\end{widetext}
Thus, in the limit  $\omega\to\infty$,% this sum becomes
\[
\lim_{\omega\to\infty}\sum{s^2_{i,j}} = \frac{9e^{2\epsilon}}{12e^{2\epsilon}}+\frac{9e^{2\epsilon+2\gamma}}{12e^{2\epsilon+2\gamma}}+\frac{e^{2\gamma}}{4e^{2\gamma}}=\frac{7}{4}
\]
in the limit  $\omega\to0$, only the constant terms (those not multiplied by $\omega$) remain.  We therefore have %Then the sum becomes
\[
\lim_{\omega\to0}\sum{s^2_{i,j}} = \frac{\mu^2 \omega_p^2}{12\mu^2 \omega_p^2}+\frac{\mu^2 \omega_p^2}{12\mu^2 \omega_p^2}+0=\frac{1}{6}
\]
This gives the desired result:
\[
\lim_{\omega\to\infty}\sum{s^2_{i,j}}=\frac{7}{4}>\lim_{\omega\to0}\sum{s^2_{i,j}}=\frac{1}{6} \ \ \ .
\]
%Putting these sums back into the equation for orthogonality we can verify that orthogonality is decreasing in the proofreading limit
%\begin{eqnarray*}
%\lim_{\omega\to0}\Theta=&1-\sqrt{2\frac{1}{6}}=0.4226\\
%>&\lim_{\omega\to\infty}\Theta=1-\sqrt{2\frac{7}{4}}=-0.8708
%\end{eqnarray*}
\subparagraph{The increase in $\sum{s^2_{i,j}}$ in monotonic}
To demonstrate that the increase in  $\sum{s^2_{i,j}}$ is monotonic in $\omega$ we must show that 
\[
\frac{d}{d\omega}\sum{s^2_{i,j}}>0
\]
We will compute the derivatives of each of the components separately.  The easiest is the $s^2_{2,3}$ term.%, let's start with this one, the derivative of this term is given as,
\begin{widetext}
\[
\frac{d}{d\omega}{s^2_{2,3}}=\frac{e^{2 \gamma } \mu  w^3 \left(e^{\gamma } w \left(3 \mu ^2+w^2+2 \mu  w\right)+e^{2 \gamma } w^2 (2 \mu +w)+\mu  \left(4 \mu ^2+2 w^2+3 \mu  w\right)\right)}{4 \left(\mu ^2+w^2+\mu  w\right)^2 \left(\mu ^2+e^{2 \gamma } w^2+e^{\gamma } \mu  w\right)^2}
\]
\end{widetext}
which is greater than zero because all rate constants are positive.  This is the desired result.

Now lets turn to the other two terms.  It is sufficient to consider the numerator of the derivatives of $\sum{s^2_{1,j}}$
\begin{widetext}
\[
\begin{aligned}
\frac{d}{d\omega}{s^2_{1,j}} &= 4 (e^\eps \oom (m'+3 \oom)-\wpp (m'-2 \oom)) [e^\eps \wpp^2 (10 m'^3+55 m'^2 \oom+39 m' \oom^2+10\oom^3)\\
   &+ e^{2 \eps} \wpp m' \oom (10 m'^2+45 m' \oom+26 \oom^2)+3 e^{3 \eps} m' \oom^3 (5 m'+\oom)+3
   \wpp^3 m' (5 m'+4 \oom)].
\end{aligned}
\]
\end{widetext}

This term is positive except for the case
\[
\begin{aligned}
\wpp m' &> \wpp 2\oom + e^\eps \oom m' + 2 e^\eps \oom\\
1 &> \frac{2\oom}{m'} + \frac{e^\eps\oom}{\wpp} + \frac{2e^\epsilon\oom^2}{\wpp m'}
\end{aligned}
\]
which is only satisfied outside of the proofreading regime $\frac{\wpp}{e^\epsilon\oom} > 1.$

\paragraph{Energetic Limit 2: $\epsilon\to\infty$}
\subparagraph{$\sum{s^2_{i,j}}$ is increasing with $\epsilon$} To demonstrate this, we will show the following.  
\[
\lim_{\epsilon\to\infty}\sum{s^2_{i,j}}>\lim_{\epsilon\to-\infty}\sum{s^2_{i,j}}
\]
First note that the term $s^2_{2,3}$ is not a function of $\epsilon$.  If we rearrange the other two $s^2_{i,j}$ terms to collect w.r.t $\epsilon$ we get,
\begin{widetext}
\begin{eqnarray*}
s^2_{1,2}&=\frac{4 w^2 \omega_p^2+e^{\epsilon } \left(12 w^3 \omega_p-2 \mu  w^2 \omega_p-2 \mu ^2 w \omega_p\right)+e^{2 \epsilon } \left(9 w^4+6 \mu  w^3+\mu ^2 w^2\right)-4 \mu  w \omega_p^2+\mu ^2 \omega_p^2}{12 w^2 \omega_p^2+e^{\epsilon } \left(16 w^3 \omega_p+16 \mu  w^2 \omega_p+16 \mu ^2 w \omega_p\right)+e^{2 \epsilon } \left(12 w^4+12 \mu  w^3+12 \mu ^2 w^2\right)+12 \mu  w \omega_p^2+12 \mu ^2 \omega_p^2}\\
s^2_{1,3}&=\frac{4 w^2 \omega_pe^{2 \gamma }+e^{\epsilon } \left(12 e^{\gamma } w^3 \omega_pe^{\gamma }-6 e^{\gamma } \mu  w^2 \omega_p+4 \mu  w^2 \omega_pe^{\gamma }-2 \mu ^2 w \omega_p\right)+e^{2 \epsilon } \left(9 e^{2 \gamma } w^4+6 e^{\gamma } \mu  w^3+\mu ^2 w^2\right)-4 \mu  w \omega_p \omega_p e^{\gamma }+\mu ^2 \omega_p^2}{12 e^{2 \gamma } w^2 \omega_p^2+e^{\epsilon } \left(16 e^{2 \gamma } w^3 \omega_p+16 w^2 \omega_p \text{$\mu $e}^{\gamma }+16 \mu ^2 w \omega_p\right)+e^{2 \epsilon } \left(12 e^{2 \gamma } w^4+12 w^3 \text{$\mu $e}^{\gamma }+12 \mu ^2 w^2\right)+12 w \omega_p^2 \text{$\mu $e}^{\gamma }+12 \mu ^2 \omega_p^2}.\\
\end{eqnarray*}
\end{widetext}

In the limit of $\epsilon\to\infty$ we have %the terms go to the limits given by the following expressions,
\begin{eqnarray*}
\lim_{\epsilon\to\infty}{s^2_{1,2}}&=\frac{\left(9 w^4+6 \mu  w^3+\mu ^2 w^2\right)}{\left(12 w^4+12 \mu  w^3+12 \mu ^2 w^2\right)}\\
\lim_{\epsilon\to\infty}{s^2_{1,3}}&=\frac{\left(9 e^{2 \gamma } w^4+6 e^{\gamma } \mu  w^3+\mu ^2 w^2\right)}{\left(12 e^{2 \gamma } w^4+12 w^3 \text{$\mu $e}^{\gamma }+12 \mu ^2 w^2\right)}.\\
\end{eqnarray*}
In contrast, as $\epsilon\to-\infty$ we have:
\begin{eqnarray*}
\lim_{\epsilon\to-\infty}{s^2_{1,2}}&=\frac{(\mu -2 w)^2}{12 \left(\mu ^2+w^2+\mu  w\right)}\\
\lim_{\epsilon\to-\infty}{s^2_{1,3}}&=\frac{\left(\mu-2 w e^{\gamma }\right)^2}{12\left(\mu ^2+e^{2 \gamma } w^2+w \text{$\mu $e}^{\gamma }\right)}.
\end{eqnarray*}
To understand the behavior of these expressions, we introduce the ratio variable $\sigma=\frac{w}{\mu}$: %making this substitution, 
\begin{eqnarray*}
\lim_{\epsilon\to\infty}{s^2_{1,2}}&=\frac{\left(9 \sigma^4+6 \sigma^3+\sigma^2\right)}{\left(12 \sigma^4+12 \sigma^3+12 \sigma^2\right)}\\
\lim_{\epsilon\to\infty}{s^2_{1,3}}&=\frac{\left(9 e^{2 \gamma } \sigma^4+6 e^{\gamma } \sigma^3+\sigma^2\right)}{\left(12 e^{2 \gamma } \sigma^4+12 \sigma^3e^{\gamma }+12\sigma^2\right)},\\
\end{eqnarray*}
and:
\begin{eqnarray*}
\lim_{\epsilon\to-\infty}{s^2_{1,2}}&=\frac{(1 -2\sigma)^2}{12 \left(1+\sigma^2+\sigma\right)}\\
\lim_{\epsilon\to-\infty}{s^2_{1,3}}&=\frac{\left(1-2 \sigma e^{\gamma }\right)^2}{12\left(1+e^{2 \gamma } \sigma^2+\sigma \mu e^{\gamma }\right)}.
\end{eqnarray*}
In the limit of large $\sigma$, we have: %both of these terms converge to 3/4 and 1/3 respectively, thus 
\[
\lim_{\epsilon\to\infty}\sum{s^2_{i,j}}\propto\frac{3}{4}>\lim_{\epsilon\to-\infty}\sum{s^2_{i,j}}\propto\frac{1}{3}
\]
\subparagraph{The increase in $\sum{s^2_{i,j}}$ is monotonic}
Again the ${s^2_{2,3}}$ term is not a function of $\epsilon$, so considering only the terms of type ${s^2_{1,j}}$
\begin{widetext}
\[
\begin{aligned}
\frac{d}{d\eps}{s^2_{1,j}} = 40 e^\eps \wpp \oom \left(m'^2+m' \oom+\oom^2\right) &\left[3 e^\eps \wpp \oom^2 (2 m'+\oom)+e^{2 \eps} m' \oom^2 (m'+3 \oom)\right.\\
&\left. +\wpp^2
   \left(-m'^2+m' \oom+2 \oom^2\right)\right]
\end{aligned}
\]
\end{widetext}

As expected, these terms are monotonically increasing except when $-m'^2\wpp^2$ dominates all other (positive) terms in the square bracket, which requires $m'$ large, and $\frac{\wpp}{e^\eps \oom}> 1,$ far from the proofreading limit.
Putting these sums back into the equation for orthogonality we can verify that orthogonality is decreasing as $\epsilon$ increases in the proofreading limit ($\sigma\approx50$)
\[
\lim_{\epsilon\to-\infty}\Theta=-0.3394>\lim_{\epsilon\to\infty}\Theta=-0.8635.
\]

\paragraph{Energetic Limit 3: $\omega_p\to0$}
Again it is instructive to rearrange $s^2_{i,j}$ to collect the $\omega_p$ terms.  Again the third term is not a function of $\omega_p$, This gives
\begin{widetext}
\begin{eqnarray*}
s^2_{1,2}&=\frac{9 w^4 e^{2 \epsilon }+6 \mu  w^3 e^{2 \epsilon }+\omega_p^2 \left(\mu ^2+4 w^2-4 \mu  w\right)+\mu ^2 w^2 e^{2 \epsilon }+\omega_p \left(12 w^3 e^{\epsilon }-2 \mu  w^2 e^{\epsilon }-2 \mu ^2 w e^{\epsilon }\right)}{12 w^4 e^{2 \epsilon }+12 \mu  w^3 e^{2 \epsilon }+\omega_p^2 \left(12 \mu ^2+12 w^2+12 \mu  w\right)+12 \mu ^2 w^2 e^{2 \epsilon }+\omega_p \left(16 w^3 e^{\epsilon }+16 \mu  w^2 e^{\epsilon }+16 \mu ^2 w e^{\epsilon }\right)}\\
s^2_{1,3}&=\frac{9 w^4 e^{2 \gamma +2 \epsilon }+6 \mu  w^3 e^{\gamma +2 \epsilon }+12 w^3 e^{\gamma +\epsilon }\omega_p e^{\gamma }+\omega_p \left(-6 \mu  w^2 e^{\gamma +\epsilon }-4 \mu  w\omega_p e^{\gamma }-2 \mu ^2 w e^{\epsilon }\right)+4 w^2\omega_p e^{2 \gamma }+4 \mu  w^2 e^{\epsilon }\omega_p e^{\gamma }+\mu ^2 w^2 e^{2 \epsilon }+\mu ^2 \omega_p^2}{12 w^4 e^{2 \gamma +2 \epsilon }+12 w^3 e^{2 \epsilon } \text{$\mu $e}^{\gamma }+\omega_p^2 \left(12 \mu ^2+12 e^{2 \gamma } w^2+12 w \text{$\mu $e}^{\gamma }\right)+12 \mu ^2 w^2 e^{2 \epsilon }+\omega_p \left(16 w^3 e^{2 \gamma +\epsilon }+16 w^2 e^{\epsilon } \text{$\mu $e}^{\gamma }+16 \mu ^2 w e^{\epsilon }\right)}\\
\end{eqnarray*}
\end{widetext}
and we must show that the sums are decreasing in $\omega_p$, i.e.
\[
\lim_{\omega_p\to0}\sum{s^2_{i,j}}>\lim_{\omega_p\to\infty}\sum{s^2_{i,j}}
\]
In the limit of $\omega_p\to0$ we have %the terms go to the limits given by the following expressions,
\begin{eqnarray*}
\lim_{\omega_p\to0}{s^2_{1,2}}&=\frac{(\mu +3 w)^2}{12 \left(\mu ^2+w^2+\mu  w\right)}\\
\lim_{\omega_p\to0}{s^2_{1,3}}&=\frac{\left(\mu +3 e^{\gamma } w\right)^2}{12 \left(\mu ^2+w \left(\text{$\mu $e}^{\gamma }+e^{2 \gamma } w\right)\right)},\\
\end{eqnarray*}
while in the limit of $\omega_p\to\infty$ 
\begin{eqnarray*}
\lim_{\omega_p\to\infty}{s^2_{1,2}}&=\frac{\mu ^2+4 w^2-4 \mu  w}{12 \mu ^2+12 w^2+12 \mu  w}\\
\lim_{\omega_p\to\infty}{s^2_{1,3}}&=\frac{\left(\mu -2 e^{\gamma } w\right)^2}{12 \left(\mu ^2+e^{2 \gamma } w^2+w \text{$\mu $e}^{\gamma }\right)}.
\end{eqnarray*}
Making the same substitutions as before ($\sigma=w/\mu$) gives:
\begin{eqnarray*}
\lim_{\omega_p\to0}{s^2_{1,2}}&=\frac{(1 +3 \sigma)^2}{12 \left(1+\sigma^2+\sigma\right)}\\
\lim_{\omega_p\to0}{s^2_{1,3}}&=\frac{\left(1 +3 e^{\gamma } \sigma\right)^2}{12 \left(1+\sigma e^{\gamma}+e^{2\gamma} \sigma^2\right)}
\end{eqnarray*}
and
\begin{eqnarray*}
\lim_{\omega_p\to\infty}{s^2_{1,2}}&=\frac{1+4 \sigma^2-4 \sigma}{12 +12 \sigma^2+12 \sigma}\\
\lim_{\omega_p\to\infty}{s^2_{1,3}}&=\frac{\left(1 -2 e^{\gamma } \sigma\right)^2}{12 \left(1+e^{2 \gamma } \sigma^2+\sigma e^{\gamma }\right)}.\\
\end{eqnarray*}
As previously we have the desired result directly:
\[
\lim_{\omega_p\to0}\sum{s^2_{i,j}}\propto\frac{3}{4}>\lim_{\omega_p\to\infty}\sum{s^2_{i,j}}\propto\frac{1}{3}.
\]
\subparagraph{The increase in $\sum{s^2_{i,j}}$ is monotonic}
We compute
\begin{widetext}
\[
\begin{aligned}
\frac{d}{d\omega_p}{s^2_{1,j}} =  -40 e^\eps \oom \left(m'^2+m' \oom+\oom^2\right) &\left[3 e^\eps \wpp \oom^2 (2 m'+\oom)+e^{2 \eps} m' \oom^2 (m'+3 \oom)\right.\\
&\left. +\wpp^2
   \left(-m'^2+m' \oom+2 \oom^2\right)\right].
\end{aligned}
\]
\end{widetext}
These terms are monotonically decreasing except for when $-m'^2\wpp^2$ dominates all other (positive) terms in the square bracket, which requires $m'$ large, and $\frac{\wpp}{e^\eps \oom}> 1.$

Putting these sums back into the equation for orthogonality we can verify that orthogonality is increasing as $\omega_p$ increases in the proofreading limit ($\sigma\approx50$)
\[
\lim_{\omega_p\to0}\Theta=-0.8635<\lim_{\omega_p\to\infty}\Theta=-0.3395
\]

\paragraph{The Hopfield Network in the Kinetic Regime}
\subparagraph{Derivation of the kinetic regime error rate}

We first derive an expression for the error rate in the kinetic regime of the Ninio-Hopfield scheme, $\xi_{\text{kinetic}}$.  We then determine the appropriate proofreading limits in the kinetic regime.  

We compute that:
\begin{widetext}
\begin{equation}\label{kin_error}
\begin{aligned}
\xi_{\text{kinetic}} &= \frac{(e^{\epsilon+\eil}\oom\oom_i+\oom\wpp+e^{\eil}\oom_i\wpp) (e^{2\delta}\oom\oom_i+e^{\delta+\epl}\oom\wpp+e^{\eil+\epl}\oom_i\wpp) }{(e^{2\delta+\eps+\eil}\oom\oom_i+e^\delta\oom\wpp+e^{\eil}\oom_i\wpp)(\oom\oom_i+e^{\epl}\oom\wpp+e^{\eil+\wpp}\oom_i\wpp)}\\
&=\frac{(e^{\eps+\eil}a + b + c)(e^{2\delta}a + e^{\epl+\delta}b + e^{\epl}c)}{(e^{2\delta+\epsilon + \eil}a + e^{\delta}b + c)(a + e^{\epl}b + e^{\epl}c)}\\
& = \frac{(e^{2\delta}a + e^{\delta+\epl}b + e^{\epl}c)(e^{\eps+\eil}a + b + c)}{(e^{2\delta+\epsilon+\eil}a + e^\delta b + c)(a + e^\epl b + e^\epl c)}
\end{aligned}
\end{equation}
\end{widetext}
where we have let $a = \oom\oom_i, \ b = \oom\wpp, \ c = \oom_i\wpp\eei.$

When the total dissipation $\eil + \epl + \eps$ is high, the terms $e^{\eps+\eil}$ in Equation \ref{kin_error} will dominate.  We therefore have that
\[
\xi_{\text{kinetic}} \approx \frac{e^{2\delta}a+e^{\delta+\epl}b + e^\epl c}{e^{2\delta}a + e^{2\delta+\epl}b + e^{2\delta+\epl}c}
\]
from which it is clear that proofreading requires that $\epp/a$ be very large.  Moreover, the error fraction is minimized when $c/b$ is very large.  Note that proofreading can still occur when $b/c>>1,$ but the error fraction is not minimized in this regime.  Translating these conditions into Kramer's form parameters gives the necessary limits for maximum discrimination
\[
\epp \to \infty, \ \frac{\oom_i\eei}{\oom} \to \infty.
\]

As in the energetic regime, we take $m'=\mu = \oom_i\eei$, and write the limits as:
\[
\epp \to \infty, \ \frac{\mu}{\omega}\to\infty \to \infty.
\]

We now investigate orthogonality in these discriminatory limits.

\subparagraph{Orthogonality is increasing with $\mu$}
Recall that increasing orthogonality requires $\sum s^2_{i,j}$ decreasing.  Lets begin by rewriting the elements of $\sum{s^2_{i,j}}$ w.r.t $\mu$
\begin{widetext}
\begin{eqnarray*}
s^2_{1,2}&=\frac{\left(\omega  e^{\delta +\epsilon } (\mu +2 \omega )+e^{-\delta_p} \omega_p(\omega -\mu )+\omega  \omega_p+\omega ^2 e^{\epsilon }\right)^2}{2 \left(\mu ^2+\mu  \omega +\omega ^2\right) \left(\left(\left(e^{\delta }+1\right) \omega  e^{\epsilon }+e^{-\delta_p} \omega_p+\omega_p\right)^2+\omega ^2 e^{2 (\delta +\epsilon )}+e^{-2 \delta_p} \omega_p^2+\omega_p^2+\omega ^2 e^{2 \epsilon }\right)}\\
s^2_{1,3}&=\frac{\left(\omega  \left(\omega  \left(-e^{\delta +\epsilon }\right)-e^{-\delta_p} \omega_p-\omega_p-\omega  e^{\epsilon }\right)+\mu  \omega_p+\omega  e^{\epsilon } (-\mu -\omega )\right)^2}{\left(\mu ^2+(\mu +\omega )^2+\omega ^2\right) \left(\left(\omega  e^{\delta +\epsilon }+e^{-\delta_p} \omega_p+\omega_p+\omega  e^{\epsilon }\right)^2+\omega ^2 e^{2 \delta +2 \epsilon }+e^{-2 \delta_p} \omega_p^2+\omega_p^2+\omega ^2 e^{2 \epsilon }\right)}\\
s^2_{2,3}&=\frac{\omega ^4 e^{2 \delta }}{4 \left(\mu ^2+\mu  \omega +\omega ^2\right) \left(\omega ^2 e^{2 \delta }+\mu  \omega  e^{\delta }+\mu ^2\right)}.\\
\end{eqnarray*}
\end{widetext}
Because $s^2_{2,3}$  has $\mu$ in the denominator but not in the numerator it must go to zero as $\mu\to\infty$.  The expressions for the remaining $s^2_{1,i}$ terms are,
\begin{widetext}
\begin{eqnarray*}
\lim_{\mu\to0}{s^2_{1,2}}&=\frac{\left(2 \omega  e^{\delta +\delta_p+\epsilon }+e^{\delta_p} \omega_p+\omega  e^{\delta_p+\epsilon }+\omega_p\right)^2}{4 \left(\omega ^2 e^{2 (\delta +\delta_p+\epsilon )}+\omega ^2 e^{\delta +2 (\delta_p+\epsilon )}+\omega  \omega_p e^{\delta +\delta_p+\epsilon }+\omega  \omega_p e^{\delta +2 \delta_p+\epsilon }+e^{\delta_p} \omega_p^2+e^{2 \delta_p} \omega_p^2+\omega ^2 e^{2 (\delta_p+\epsilon )}+\omega  \omega_p e^{\delta_p+\epsilon }+\omega  \omega_p e^{2 \delta_p+\epsilon }+\omega_p^2\right)}\\
\lim_{\mu\to0}{s^2_{1,3}}&=\frac{\left(\omega  e^{\delta +\delta_p+\epsilon }+e^{\delta_p} \omega_p+2 \omega  e^{\delta_p+\epsilon }+\omega_p\right)^2}{4 \left(\omega ^2 e^{2 (\delta +\delta_p+\epsilon )}+\omega ^2 e^{\delta +2 (\delta_p+\epsilon )}+\omega  \omega_p e^{\delta +\delta_p+\epsilon }+\omega  \omega_p e^{\delta +2 \delta_p+\epsilon }+e^{\delta_p} \omega_p^2+e^{2 \delta_p} \omega_p^2+\omega ^2 e^{2 (\delta_p+\epsilon )}+\omega  \omega_p e^{\delta_p+\epsilon }+\omega  \omega_p e^{2 \delta_p+\epsilon }+\omega_p^2\right)}\\
\end{eqnarray*}
and
\begin{eqnarray*}
\lim_{\mu\to\infty}{s^2_{1,2}}&=\frac{\left(\omega_p-\omega  e^{\delta +\delta_p+\epsilon }\right)^2}{4 \left(\omega ^2 e^{2 (\delta +\delta_p+\epsilon )}+\omega ^2 e^{\delta +2 (\delta_p+\epsilon )}+\omega  \omega_p e^{\delta +\delta_p+\epsilon }+\omega  \omega_p e^{\delta +2 \delta_p+\epsilon }+e^{\delta_p} \omega_p^2+e^{2 \delta_p} \omega_p^2+\omega ^2 e^{2 (\delta_p+\epsilon )}+\omega  \omega_p e^{\delta_p+\epsilon }+\omega  \omega_p e^{2 \delta_p+\epsilon }+\omega_p^2\right)}\\
\lim_{\mu\to\infty}{s^2_{1,3}}&=\frac{e^{2 \delta_p} \left(\omega_p-\omega  e^{\epsilon }\right)^2}{4 \left(\omega ^2 e^{2 (\delta +\delta_p+\epsilon )}+\omega ^2 e^{\delta +2 (\delta_p+\epsilon )}+\omega  \omega_p e^{\delta +\delta_p+\epsilon }+\omega  \omega_p e^{\delta +2 \delta_p+\epsilon }+e^{\delta_p} \omega_p^2+e^{2 \delta_p} \omega_p^2+\omega ^2 e^{2 (\delta_p+\epsilon )}+\omega  \omega_p e^{\delta_p+\epsilon }+\omega  \omega_p e^{2 \delta_p+\epsilon }+\omega_p^2\right)}.\\
\end{eqnarray*}
\end{widetext}
Here we will again make a ratio substitution, $\tau=\omega e^{\epsilon}/\omega_p$ and send $\tau\to\infty$ and required in the kinetic discriminatory regime.  In this limit, we have: %get the following expressions
\begin{eqnarray*}
\lim_{\mu\to0}{s^2_{1,2}}&=\frac{1+4e^{\delta}+4e^{2\delta}}{4+4e^{\delta}+4e^{2\delta}}\\
\lim_{\mu\to0}{s^2_{1,3}}&=\frac{4+4e^{\delta}+e^{2\delta}}{4+4e^{\delta}+4e^{2\delta}}\\
\end{eqnarray*}
and
\begin{eqnarray*}
\lim_{\mu\to\infty}{s^2_{1,2}}&=\frac{e^{2\delta}}{4+4e^{\delta}+4e^{2\delta}}\\
\lim_{\mu\to\infty}{s^2_{1,3}}&=\frac{1}{4+4e^{\delta}+4e^{2\delta}}.\\
\end{eqnarray*}

This gives the desired result,
\[
\lim_{\mu\to0}\sum{s^2_{i,j}}>\lim_{\mu\to\infty}\sum{s^2_{i,j}}.
\]
Putting these sums back into the equation for orthogonality we can verify that orthogonality is increasing as $\mu$ increases in the proofreading limit ($\tau\approx10^4$):
\[
\lim_{\mu\to0}\Theta=-0.819<\lim_{\mu\to\infty}\Theta=0.3612.
\]
\paragraph{The increase in $\sum{s^2_{i,j}}$ is monotonic}
Again it is easiest to start with the $s^2_{2,3}$ term.  An application of the quotient rule reveals that $\frac{d}{d\mu}s^2_{2,3}<0.$  

The remaining derivatives are given by
\begin{widetext}
\begin{eqnarray*}
&\frac{d}{d\mu}s^2_{1,2}=\\
&-\frac{\omega  \left(\omega  (\mu +2 \omega ) e^{\delta +\delta_p+\epsilon }+e^{\delta_p} \omega  \omega_p+\omega ^2 e^{\delta_p+\epsilon }+\omega_p (\omega -\mu )\right) \left(3 \mu  \omega  e^{\delta +\delta_p+\epsilon }+e^{\delta_p} \omega_p (2 \mu +\omega )+\omega  e^{\delta_p+\epsilon } (2 \mu +\omega )+3 \omega_p (\mu +\omega )\right)}{4 \left(\mu ^2+\mu  \omega +\omega ^2\right)^2 \left(\omega ^2 e^{2 (\delta +\delta_p+\epsilon )}+\omega ^2 e^{\delta +2 (\delta_p+\epsilon )}+\omega  \omega_p e^{\delta +\delta_p+\epsilon }+\omega  \omega_p e^{\delta +2 \delta_p+\epsilon }+e^{\delta_p} \omega_p^2+e^{2 \delta_p} \omega_p^2+\omega ^2 e^{2 (\delta_p+\epsilon )}+\omega  \omega_p e^{\delta_p+\epsilon }+\omega  \omega_p e^{2 \delta_p+\epsilon }+\omega_p^2\right)}\\
&\frac{d}{d\mu}s^2_{1,3}=\\
&-\frac{\omega  \left(\omega  (2 \mu +\omega ) e^{\delta +\delta_p+\epsilon }+3 e^{\delta_p} \omega_p (\mu +\omega )+3 \mu  \omega  e^{\delta_p+\epsilon }+\omega_p (2 \mu +\omega )\right) \left(\omega  e^{\delta_p+\epsilon } \left(\left(e^{\delta }+2\right) \omega +\mu \right)+\omega_p \left(e^{\delta_p} (\omega -\mu )+\omega \right)\right)}{4 \left(\mu ^2+\mu  \omega +\omega ^2\right)^2 \left(\omega ^2 e^{2 (\delta +\delta_p+\epsilon )}+\omega  e^{\delta +\delta_p+\epsilon } \left(e^{\delta_p} \left(\omega_p+\omega  e^{\epsilon }\right)+\omega_p\right)+e^{\delta_p} \omega_p^2+e^{2 \delta_p} \omega_p^2+\omega ^2 e^{2 (\delta_p+\epsilon )}+\omega  \omega_p e^{\delta_p+\epsilon }+\omega  \omega_p e^{2 \delta_p+\epsilon }+\omega_p^2\right)}\\
\end{eqnarray*}
\end{widetext}
Which are both strictly negative.  We conclude that $\sum{s^2_{i,j}}$ is a monotonically decreasing function of $\mu,$ thus orthogonality is monotonically increasing.

\section{Expressions for Error Rate in the Ladder Graph}
\label{app:ladder_error}
We wish to derive expressions for the error rate of the ladder discrimination scheme in the kinetic and energetic regimes.

A single side of the ladder has structure:
\begin{center}
\schemestart
0 \arrow(0--$y_{s0}$){<=>[$\mathrm{k_{on}}$][$\mathrm{k_{off}}$]}[,,,red] $y_{s0}$
\arrow{<=>[$u$][$d$]}[90,,,red] 
$x_{s0}$ \arrow($x_{s0}$--$x_{s1}$){->[$f$]}[,,,red]$x_{s1}$
\arrow(@$y_{s0}$--$y_{s1}$){<-[][$b$]}[,,,red] $y_{s1}$ \arrow{<=>[$u$][$d$]}[90,,,red]
\arrow(@$y_{s1}$--$y_{s2}$){<-[][$b$]}[,,,red] $y_{s2}$
\arrow{<=>[$u$][$d$]}[90,,,red] 
\arrow(@$x_{s1}$--$x_{s2}$){->[$f$]}[,,,red] $x_{s2}$
\schemestop
\end{center}
where we have dropped the superscripts $d^S, \ u^S$ for clarity.

We will use the Matrix-Tree theorem (MTT), which provides an expression for steady states in terms of {\it spanning trees}~\cite{Wong2018-ys}.  Recall that a {\it spanning tree} of a graph $G$ is is a subgraph which: includes every vertex of $G$ and has no cycles (when edge directions ignored).  A spanning tree is said to be {\it rooted} at node $i$ if $i$ is the only vertex of the subgraph without any outgoing edges.

The MTT provides an expression for the steady states of node $i$ in terms of the sum of the product of the rates of each spanning tree rooted at $i$.  That is:
\[
\rho_i = \sum_{T\in S_i(G) }\left(\prod_{j\stackrel{a}{\to} k\in T}a \right),
\]
where $S_i(G)$ is the set of all spanning trees of graph $G$ rooted at $i$.

We will exploit the structure of our ladder network in order to simplify this expression.  Our ladder consists of two subgraphs (call them $G_R, \ G_W,$ corresponding to right, wrong products, respectively).  These two subgraphs are joined at a single node, $0.$  Because the subgraphs share a single node, the kernel element corresponding to node $i$ in subgraph $R$ is given by $\rho_i = \rho(G_R)\rho_0(G_W)$~\cite{Wong2018-ys}.

This gives for the error
\begin{equation}\label{jer_span}
\xi = \frac{\rho_W}{\rho_R} = \frac{\rho_W(G_W)\rho_0(G_R)}{\rho_R(G_R)\rho_0(G_W)}.
\end{equation}

We know that $\rho_S(G_S)$ $(\{S=W,R\})$ in Equation \ref{jer_span} represents the node at the top corner of the graph.  Similarly, $\rho_0(G_S)$ represents the node 0. 

We therefore need only determine analytical expressions for the sums of (products of rate constants of) spanning trees rooted at the top corner and 0 nodes. 

Let's count the trees rooted at $\rho_0(G_S)$ first.  In order for the tree to be rooted at $0$, there are a number of essential arrows:

\bigskip

\begin{center}
\schemestart
0 \arrow(0--$y_{s0}$){<-[][$\mathrm{k_{off}}$]}[,,,red] $y_{s0}$
\arrow{<=>[$u$][$d$]}[90,,,white] 
$x_{s0}$ \arrow($x_{s0}$--$x_{s1}$){->[$f$]}[,,,white]$x_{s1}$
\arrow(@$y_{s0}$--$y_{s1}$){<-[][$b$]}[,,,red] $y_{s1}$ \arrow{<=>[$u$][$d$]}[90,,,white]
\arrow(@$y_{s1}$--$y_{s2}$){<-[][$b$]}[,,,red] $y_{s2}$
\arrow{<-[][$d$]}[90,,,red] 
\arrow(@$x_{s1}$--$x_{s2}$){->[$f$]}[,,,white] $x_{s2}$
\schemestop
\end{center}
without any of which it is impossible to produce a spanning tree rooted at $0.$  The necessity of these arrows comes from the unidirectionality of the $f, b.$

What other arrows are necessary for a spanning tree?  Consider the diagram

\begin{center}
\schemestart
0 \arrow(0--$y_{s0}$){<-[][$\mathrm{k_{off}}$]}[,,,red] $y_{s0}$
\arrow{<-[][$d$]}[90,,dashed,green] 
$x_{s0}$ \arrow($x_{s0}$--$x_{s1}$){->[$f$]}[,,,green]$x_{s1}$
\arrow(@$y_{s0}$--$y_{s1}$){<-[][$b$]}[,,,red] $y_{s1}$ \arrow{<-[][$d$]}[90,,dashed,blue]
\arrow(@$y_{s1}$--$y_{s2}$){<-[][$b$]}[,,,red] $y_{s2}$
\arrow{<-[][$d$]}[90,,,red] 
\arrow(@$x_{s1}$--$x_{s2}$){->[$f$]}[,,,blue] $x_{s2}$
\schemestop
\end{center}

It is necessary and sufficient for a spanning tree rooted at $0$ to contain all of the red arrows, and exactly one of the green arrows and one of the blue arrows.  This holds in general; each loop in a ladder must contribute either a factor of $f$ or $d$ to a spanning tree rooted at 0.

We can thus count the number of possible spanning trees, and the product of their rate constants
\[
\begin{aligned}
\rho_0(G_S) &= k_{\rm off}b^\alpha d \sum_{k=0}^\alpha {\alpha\choose k}f^{\alpha-k}d^k\\
&=k_{\rm off}b^\alpha d (f + d)^\alpha
\end{aligned}
\]
where the second line follows from the Binomial theorem, and where we have set the number of square loops in the ladder portion of the graph to be $\alpha.$

We can now repeat this procedure with spanning trees rooted in the upper corner, with red, blue, and green as before:

\begin{center}
\schemestart
0 \arrow(0--$y_{s0}$){->[$\mathrm{k_{on}}$][]}[,,,red] $y_{s0}$
\arrow{->[$u$][$$]}[90,,,red] 
$x_{s0}$ \arrow($x_{s0}$--$x_{s1}$){->[$f$]}[,,,red]$x_{s1}$
\arrow(@$y_{s0}$--$y_{s1}$){<-[][$b$]}[,,dashed,green] $y_{s1}$ \arrow{->[$u$][$$]}[90,,dashed,green]
\arrow(@$y_{s1}$--$y_{s2}$){<-[][$b$]}[,,dashed,blue] $y_{s2}$
\arrow{->[$u$][$$]}[90,,dashed,blue] 
\arrow(@$x_{s1}$--$x_{s2}$){->[$f$]}[,,,red] $x_{s2}$
\schemestop
\end{center}

Which gives us:
\[
\begin{aligned}
\rho_S(G_S) &= k_{\rm on}f^\alpha u \sum_{k=0}^\alpha {\alpha\choose k}b^{\alpha-k}u^k\\
&=k_{\rm on}f^\alpha u (b + u)^\alpha.
\end{aligned}
\]
Note that in comparison to the last expression, we have merely made the substitutions: $b\to f, \ f \to b, \ d \to u, \  u \to d.$  Plus $k_{\rm off}\to k_{\rm on},$ of course.

Returning to our expression for the error gives
\[
\begin{aligned}
\xi &= \frac{\rho_W}{\rho_R} = \frac{\rho_W(G_W)\rho_0(G_R)}{\rho_R(G_R)\rho_0(G_W)}\\
&= \frac{k_{\rm on}k_{\rm off}\  f_W^\alpha \ u_W \ b_R^\alpha \  d_R \  (u_W + b_W)^\alpha(f_R+d_R)^\alpha }{k_{\rm on}k_{\rm off}\  f_R^\alpha \ u_R \ b_W^\alpha \  d_W \  (u_R + b_R)^\alpha(f_W+d_W)^\alpha }
\end{aligned}
\]
where we have denoted variables coming from the `right' and `wrong' sides of the ladder with subscripts $R$ and $W,$ respectively.  We can do some cancellation (b = $b_R$ = $b_W$ and $f = f_R = f_W$) to arrive at:
\[
\xi = \frac{d_Ru_W(u_W+b)^{\alpha}(f+d_R)^\alpha}{d_Wu_R(u_R+b)^{\alpha}(f+d_W)^\alpha}.
\]
In the energetic regime we have that $u_R = u_W,$ and that $d_W = d_Re^\gamma:$
\[
\xi_{\rm energetic} = \frac{(f+d_R)^\alpha}{e^\gamma(f+d_Re^\gamma)^\alpha}.
\]
In the kinetic regime, we have that $d_R = d_We^\delta, \ u_R = u_We^\delta,$ giving 
\[
\xi_{\rm kinetic} =  \frac{(u+b)^\alpha(f+de^\delta)^\alpha}{(ue^\delta+b)^\alpha(f+d)^\alpha}.
\]

\section{Orthogonality and Error in the Ladder Graph}
\label{app:ladder_orth}

We first derive the discriminatory limit in the energetic regime.  Recall that

\subsection{Energetic regime}
\[
\xi_{\rm energetic} = \frac{(f+d_R)^\alpha}{e^\gamma(f+d_Re^\gamma)^\alpha}.
\]
The substitution $\eta=\frac{d_R}{f}$ gives
\[
\xi_{\rm energetic} = \frac{(1+\eta)^\alpha}{e^\gamma(1+\eta e^\gamma)^\alpha}
\]
from which read off that proofreading requires $\eta$ to be large.  This corresponds to the intuition that the rate of discards must be large with respect the reaction speed.  

We must now demonstrate that orthogonality is decreasing as $\eta$ becomes large.

As in the Ninio-Hopfield case, we will use the notation $\sum{s^2_{i,j}}$ to denote the squared, normalized inner product between columns $i, j$ in Matrix $\mathcal{L}^{a,b}$ formed by deleting the columns corresponding to the discriminatory nodes $a, b$ from the full Laplacian for this graph.

For any given loop of the ladder, these terms are given by
\[
\begin{aligned}
 \la x_{si}, y_{s(i+1)} \ra^2 &=\frac{(b d+f u)^2}{4 \left(b^2+b u+u^2\right) \left(d^2+d f+f^2\right)}\\
 \la y_{si}, y_{s(i+1)} \ra^2 &= \frac{b^2 (b+u)^2}{4 \left(b^2 + bu+ u^2\right)^2}\\
 \la x_{si}, x_{s(i+1)}\ra^2 &=\frac{f^2 (d+f)^2}{4 \left(d^2+d f+f^2\right)^2}\\
 \la x_{si}, y_{si}\ra^2 &=\frac{(b d+2 d u+f u)^2}{4 \left(b^2+b u+u^2\right) \left(d^2+d f+f^2\right)}.
 \end{aligned}
\]
For $N$ loops, there will be $N$ of each of these terms except for $ \la x_{si}, x_{s(i+1)}\ra^2$ for which there will be $(N-1)$ for each side of the ladder.  In addition, there will be two terms that originate from the reactant node (note in this case we are considering a slightly altered graph, where $k_{off}=b$ and $k_{on}=f$, and $k_{on}$ connects 0 to $x_{s0}$).   These are given as
\[
\begin{aligned}
\la 0, x_{s0} \ra^2 &= \frac{(2 b-u)^2}{12 \left(b^2+b u+u^2\right)} \\
\la 0, y_{s0} \ra^2 &= \frac{(d+f)^2}{6 \left(d^2+(d+f)^2+f^2\right)}.\\
\end{aligned}
\]
Recall that in the energetic regime, our effective parameter of interest if $\eta=d/f$, noting that  $\la y_{si}, y_{s(i+1)}$ and $\la 0, x_{s0} \ra^2$ are not functions of $\eta$ and making this substitution along with the another substitute $\phi=u/b$ gives %the following expressions,
\[
\begin{aligned}
 \la x_{si}, y_{s(i+1)} \ra^2 &=\frac{(\eta+\phi)^2}{4 \left(1+ \phi+\phi^2\right) \left(1+\eta+\eta^2\right)}\\
 \la x_{si}, x_{s(i+1)}\ra^2 &=\frac{(\eta+1)^2}{4 \left(1+\eta+\eta^2\right)^2}\\
 \la x_{si}, y_{si}\ra^2 &=\frac{(\eta+2 \eta \phi+ \phi)^2}{4 \left(1+\phi+\phi^2\right) \left(1+\eta+\eta^2\right)}\\
 \la 0, y_{s0} \ra^2 &=\frac{(\eta+1)^2}{12 \left(1+\eta+\eta^2\right)}\\
 \end{aligned}
\]
We will set $\phi\to0$ for convenience.  In this limit we have:
%What values should $\phi$ take for efficient proofreading in the energetic regime?  $u$ is the rate of rescues this rate must be small relative to the rate of backtracking for efficient proofreading which implies that $\phi\to 0$.  In this limit, the expressions become.
\[
\begin{aligned}
\la x_{si}, y_{s(i+1)} \ra^2 &=\frac{\eta^2}{4 \left(1+\eta+\eta^2\right)}\\
\la x_{si}, x_{s(i+1)}\ra^2 &=\frac{(\eta+1)^2}{4 \left(1+\eta+\eta^2\right)^2}\\
\la x_{si}, y_{si}\ra^2 &=\frac{\eta^2}{4 \left(1+\eta+\eta^2\right)}\\
\la 0, y_{s0} \ra^2 &=\frac{(\eta+1)^2}{12 \left(1+\eta+\eta^2\right)}\\
\end{aligned}
\]
which take values 0, 1/4, 0, and 1/12 in the limit $\eta\to0$ and 1/4, 0, 1/4, and 1/12 in the limit $\eta\to\infty$ For $N$ loops, we will have N terms of the first and third type, and $N-1$ terms of the second type.  The last term is unchanged in these limits.  This gives the desired result, 
\[
\lim_{\eta\to 0}\sum s^2_{i,j} \propto \frac{N-1}{4}<\lim_{\eta\to\infty}\sum s^2_{i,j} \propto \frac{2N}{4}.
\]
Finally, we consider the case when $\phi\to\infty$.  Note that $\la x_{si}, x_{s(i+1)}\ra^2$ terms are not functions of $\phi$.  The two remaining terms to consider are,
\[
\begin{aligned}
 \la x_{si}, y_{s(i+1)} \ra^2 &=\frac{(\eta+\phi)^2}{4 \left(1+ \phi+\phi^2\right) \left(1+\eta+\eta^2\right)}\\
 \la x_{si}, y_{si}\ra^2 &=\frac{(\eta+2 \eta \phi+ \phi)^2}{4 \left(1+\phi+\phi^2\right) \left(1+\eta+\eta^2\right)}\\
 \end{aligned}
\]
which in the $\phi\to\infty$ limit become,
\[
\begin{aligned}
 \la x_{si}, y_{s(i+1)} \ra^2 &=\frac{1}{4 \eta ^2+4 \eta +4}\\
 \la x_{si}, y_{si}\ra^2 &=\frac{4 \eta ^2+4 \eta +1}{4 \eta ^2+4 \eta +4}\\
 \end{aligned}
\]
Combining these term yields,
\[
\lim_{\eta\to 0}\sum s^2_{i,j} \propto \frac{2N-1}{4}<\lim_{\eta\to\infty}\sum s^2_{i,j} \propto \frac{4N}{4}.
\]

We can directly compute that $\sum s^2_{i,j} $ is monotonically increasing in both the $\phi\to0$ and $\phi\to\infty$ limits.

\subsubsection{$\phi$ does not affect orthogonality in the $f\ll d$ limit}

Before turning to the kinetic regime, we demonstrate that $\phi$ does not affect orthogonality in the energetic discrimination limit.

We examine the elements $s^2_{i,j}$ that depend on $\phi$ in the $\eta\to\infty$ discriminatory limit.  Before taking the limit, we have
\[
\begin{aligned}
\la x_{si}, y_{s(i+1)} \ra^2 &=\frac{(\eta+\phi)^2}{4 \left(1+ \phi+\phi^2\right) \left(1+\eta+\eta^2\right)}\\
\la x_{si}, y_{si}\ra^2 &=\frac{(\eta+2 \eta \phi+ \phi)^2}{4 \left(1+\phi+\phi^2\right) \left(1+\eta+\eta^2\right)}\\
\la y_{si}, y_{s(i+1)}\ra^2 &=\frac{(\phi +1)^2}{4 \left(\phi ^2+\phi +1\right)^2}.\\
\end{aligned}
\]
In the $\eta\to\infty$ limit these become
\[
\begin{aligned}
\la x_{si}, y_{s(i+1)} \ra &=\sqrt{\frac{1}{4 \phi ^2+4 \phi +4}}\\
\la x_{si}, y_{si}\ra &=\sqrt{\frac{4 \phi ^2+4 \phi +1}{4 \phi ^2+4 \phi +4}}\\
\la y_{si}, y_{s(i+1)}\ra &=\frac{(\phi +1)}{2 \left(\phi ^2+\phi +1\right)}.\\
\end{aligned}
\]
Now we must evaluate  these in the limits of $\phi\to 0$ and $\phi\to\infty$, the first term goes from 1/2 to 0 as $\phi\to\infty$.  The second term goes from 1/2 to 1 and the third term goes from 1/2 to 0.  Because each loop consists of two of the second type term and one each of the first and third type term, the sum is the same in each limit.  

In the full expression for orthogonality, we do observe a small non-constant dependence on $\phi$, but this is marginal and strictly decreases the orthogonality, thereby reinforcing our notion that $\phi$ cannot be used to increased realizable pathways in the low $f$ regime.  

\subsection{Kinetic regime}

We now need to demonstrate that 

\[
\xi_{\rm kinetic} =  \frac{(u+b)^\alpha(f+de^\delta)^\alpha}{(ue^\delta+b)^\alpha(f+d)^\alpha}.
\]
Define $\eta = d/f, \ \phi = u/b$ as before.
\[
\xi_{\rm kinetic} =  \frac{(\phi+1)^\alpha(1+\eta e^\delta)^\alpha}{(\phi e^\delta+1)^\alpha(\eta+1)^\alpha}.
\]
which attains its minimum of $e^{-\alpha\delta}$ in the limit $\phi \to \infty, \ \eta \to 0.$  
The previous sections demonstrated that orthogonality is increasing in these limits.

\clearpage
\section{Supplemental Information}
\label{sec:supplement}
\renewcommand{\thefigure}{S\arabic{figure}}
\setcounter{figure}{0}

\begin{figure}[htb]
\includegraphics[width = 0.4\textwidth]{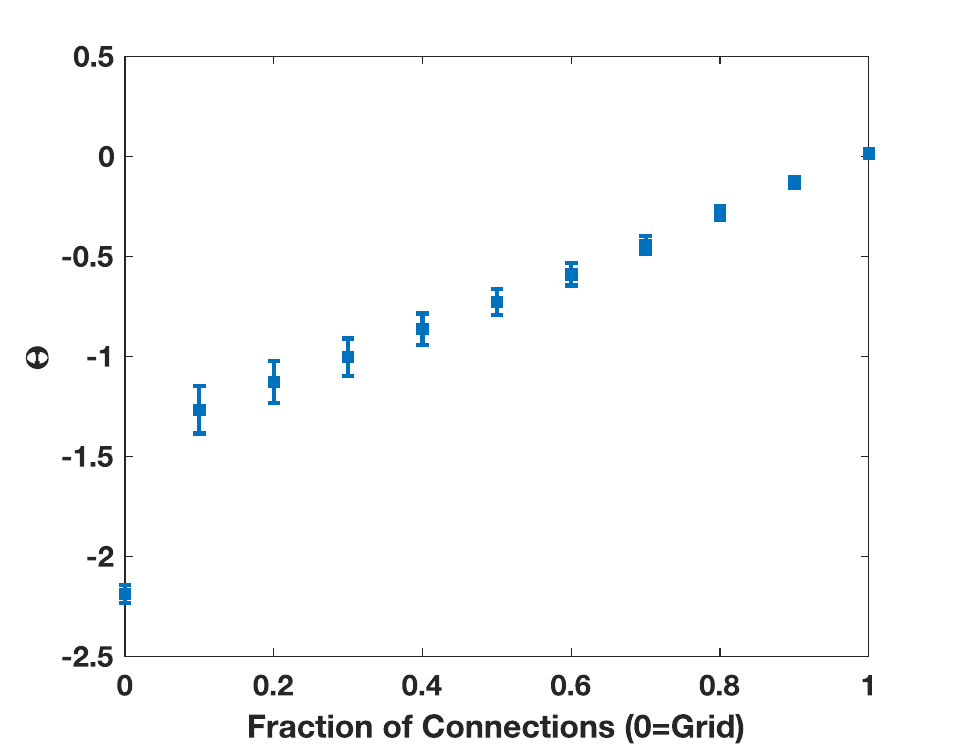}
\caption{As a graph becomes more connected, orthogonality increases.  Orthogonality is plotted against varying connectivities of a 16 node graph, generated as described in the main text of this section.  Zero connectivity corresponds to a $16\times 16$ grid graph, by convention.}
\label{sfig:orth}
\end{figure}

\begin{figure}[htb]
\includegraphics[width = 0.4\textwidth]{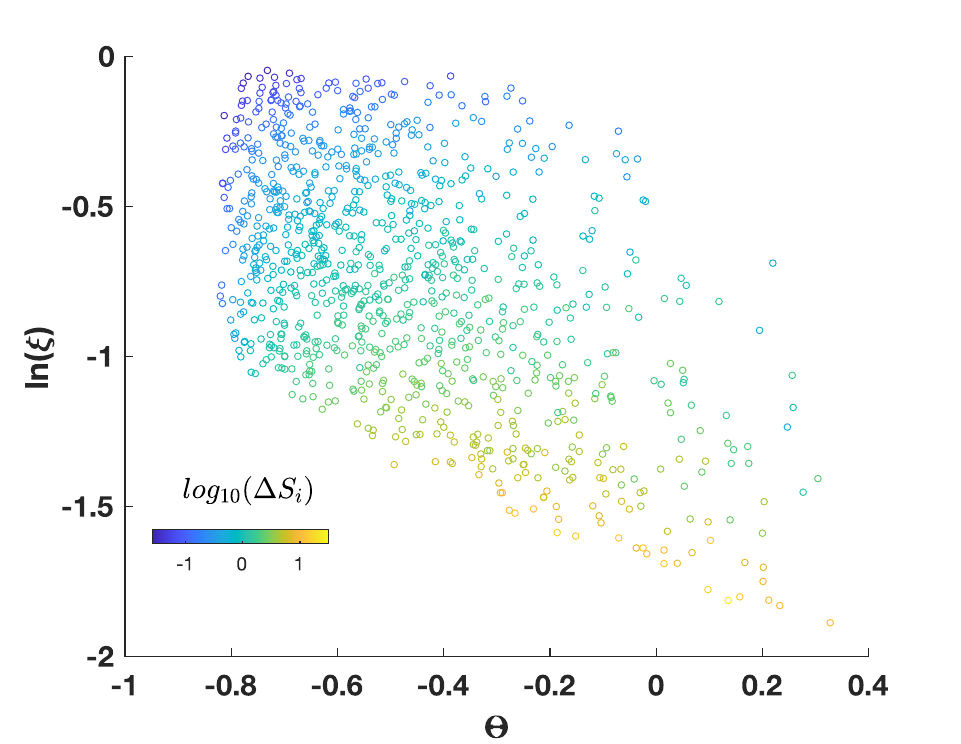}
\caption{Orthogonality is required to achieve the minimum error rate in the kinetic regime ($\gamma=0, \delta=1$). The log of the error rate $(\ln(\xi))$ as a function of the orthogonality ($\Theta$) is plotted for simulations of the triangle graph (Hopfield-Ninio) with Kramer's form rate constants for 1,000 randomly chosen values of $\omega_i$, $\omega_p$, $\epsilon_i$, and $\epsilon_p$.   Other parameters were fixed  ($\omega=1$, $\epsilon=10$).   Color shows the dissipation $\Delta S_i$ at steady state. \label{sfig:kin_sim}}
\end{figure}

Figure \ref{sfig:orth} demonstrates that orthogonality tends to increase as we add connections of equal order of magnitude to a graph.  The figure was generated by first creating an all-to-all connected graph having 16 nodes.  Rate constants were chosen randomly from the distribution $Exp[\mathcal{N}(0,\frac{1}{3})]$.  For each of the $c = \text{`connectivity fractions'}$ in Figure \ref{sfig:orth}, a random set of $1 - c*(16^2-16)/2$ connections was then chosen for deletion and removed bidirectionally.  These random deletion sets were chosen 1000 times for each connectivity fraction considered.  The mean and standard deviation of these 1000 samples is plotted.  Graph sparsity 0 corresponds to a $16\times16$ grid graph.

Figure \ref{sfig:kin_sim} shows the relationship between orthogonality and error for the Ninio-Hopfield model in the kinetic regime $(\gamma=0,\delta=\delta_p=1)$.  High orthogonality and high dissipation are necessary for low error.

\subsection{Tables of parameter values}
Table \ref{tab:irr_ladder} gives the values of the rate constants in the irreversible style ladder graph model, and used to generate the plots in Figure \ref{fig:ladder}
\begin{table}[htbp]
\begin{center}
\begin{tabular}{|c|c|c|}
\hline
Parameter & Fig\ref{fig:ladder}(b) Energetic & Fig\ref{fig:ladder}(b) Kinetic\\
\hline
$f$&0.1&2\\
$b$&2&0.1\\
$u$&0.1&3\\
$d$&$f(x)$&$f(x)$\\
$\gamma$&1&0\\
$\delta$&0&1\\
\hline
\end{tabular}
\end{center}
\caption{Parameters used to generate different figures for the irreversible ladder graph.  $f(x)$ indicate that this parameter was used as an independent variable for plotting. \label{tab:irr_ladder}}
\end{table}%
Table \ref{tab:rev_ladder} gives the values used to derive Kramer's form rate constants for the reversible ladder graph and to generate the plots shown in Figure \ref{fig:ladder}.
\begin{table}[htbp]
\begin{center}
\begin{tabular}{|c|c|c|}
\hline
Parameter & Fig\ref{fig:ladder_switch}(c) & Fig\ref{fig:ladder_switch}(d)\\
\hline
$\omega_f$&$f(x)$&$0.0874$\\
$\omega_b$&2.3565&2.3565\\
$\omega_d$&15.33&$f(x)$\\
$\epsilon_f$&3&$3$\\
$\epsilon_b$&3&3\\
$\epsilon_u$&3&3\\
$\gamma$&1&1\\
$\delta$&0&0\\
\hline
\end{tabular}
\end{center}
\caption{Parameters used to generate different figures for the reversible ladder graph.  $f(x)$ indicate that this parameter was used as an independent variable for plotting.  \label{tab:rev_ladder}}
\end{table}%
\begin{table*}[hptb]
\begin{center}
\begin{tabular}{|c|c|c|c|c|c|}
\hline
Parameter & Fig\ref{fig:hopfield}(b) & Fig\ref{fig:hopfield}(c)& Fig\ref{fig:hopfield}(d) Energetic&Fig\ref{fig:hopfield}(d) Kinetic&Fig \ref{sfig:kin_sim}\\
\hline
$\omega$&1&1&1&1&1\\
$\epsilon$&10&10&10&10&10\\
$\gamma$&1&1&1&0&0\\
$\delta$&0&0&0&1&1\\
$\delta_p$&0&0&0&1&1\\
$\omega_i$&$Exp[\mathcal{N}(0,\frac{1}{2})]$&0.55&0.55&2.27&$Exp[\mathcal{N}(0,\frac{1}{2})]$\\
$\epsilon_i$&$\mathcal{N}(0,2)$&$f(x)$&$f(x)$&$f(x)$&$\mathcal{N}(0,2)$\\
$\omega_p$&$Exp[\mathcal{N}(0,\frac{1}{2})]$&0.7318&0.7318&0.8982&$Exp[\mathcal{N}(0,\frac{1}{2})]$\\
$\epsilon_p$&$\mathcal{N}(0,2)$&4.2245&4.2245&2.5553&$\mathcal{N}(0,2)$\\
\hline
\end{tabular}
\end{center}
\caption{Parameters used to generate different figures for the Hopfield-Ninio Model.  $f(x)$ indicate that this parameter was used as a dependent variable for plotting.  \label{tab:hopfield}}
\end{table*}%
Table \ref{tab:hopfield} gives the values used to derive Kramer's form rate constants for the Hopfield-Ninio model and to generate the plots shown in Figure \ref{fig:hopfield}.
\end{document}